\documentclass[11pt,letterpaper]{article}

\usepackage[T1]{fontenc}
\usepackage{textcomp}
\usepackage[centertags]{amsmath}
\usepackage{amsfonts}
\usepackage{amsthm}
\usepackage{amssymb}
\usepackage{graphicx}
\usepackage{multirow}
\usepackage[numbers,sort&compress]{natbib}
\usepackage[pdftex]{hyperref}
\usepackage{calc}

\usepackage{paperinitial}

\paperinitialization{15mm}{15mm}{20mm}{10mm}{2pt}{10pt}

\DeclareMathOperator{\re}{Re}
\DeclareMathOperator{\im}{Im}
\DeclareMathOperator{\Sp}{Sp}
\DeclareMathOperator{\Texp}{Texp}
\DeclareMathOperator{\sgn}{sgn}
\DeclareMathOperator{\tah}{th}
\DeclareMathOperator{\ch}{ch}
\DeclareMathOperator{\sh}{sh}
\DeclareMathOperator{\arctg}{arctg}
\DeclareMathOperator{\arsh}{arsh}
\DeclareMathOperator{\arch}{arch}
\DeclareMathOperator{\arth}{arth}
\DeclareMathOperator{\arcth}{arcth}
\newcommand{\lan}{\langle}
\newcommand{\ran}{\rangle}

\newcommand{\e}{\varepsilon}
\newcommand{\vf}{\varphi}

\newcommand{\s}{\sigma}

\newcommand{\Si}{\Sigma}
\newcommand{\al}{\alpha}
\newcommand{\be}{\beta}
\newcommand{\ga}{\gamma}
\newcommand{\Ga}{\Gamma}
\newcommand{\de}{\delta}
\newcommand{\De}{\Delta}

\newcommand{\la}{\lambda}
\newcommand{\La}{\Lambda}
\newcommand{\ups}{\upsilon}

\newcommand{\spx}{\mathbf{x}}
\newcommand{\spy}{\mathbf{y}}

\newtheorem{thm}{Theorem}
\theoremstyle{definition}
\newtheorem{defn}{Definition}

\begin{document}
\setlength{\unitlength}{1pt}
\allowdisplaybreaks[4]


\title{{\Large \textbf{Large mass expansion of the one-loop effective action induced\\
by a scalar field on the two-dimensional Minkowski background\\[-4pt]
with non-trivial $(1+1)$ splitting}}}

\date{}

\author{P.O. Kazinski\thanks{E-mail: \texttt{kpo@phys.tsu.ru}},\; V.D. Miller\thanks{E-mail: \texttt{nl\_mvd@mail.ru}}\\[0.5em]
{\normalsize Physics Faculty, Tomsk State University, Tomsk 634050, Russia}}

\maketitle

\begin{abstract}

A large mass expansion of the one-loop effective action of a scalar field on the two-dimensional Minkowski spacetime is found in the system of coordinates, where the metric $g_{\mu\nu}(t,x)\neq\eta_{\mu\nu}=diag(1,-1)$, and $g_{\mu\nu}(t,x)$ tends to $\eta_{\mu\nu}$ at the spatial and temporal infinities. It is shown that, apart from the Coleman-Weinberg potential, this expansion contains the terms both analytic and non-analytic in $m^{-2}$, where $m$ is the mass of a scalar field. A general unambiguous expression for the one-loop correction to the effective action on non-stationary backgrounds is derived.

\end{abstract}


\section{Introduction}

The evaluation of one-loop corrections to the effective action (the generating functional of one-particle irreducible Green functions) on the background of averages of quantum fields is a rather developed branch of theoretical physics. In order to find the one-loop correction, one needs to evaluate a certain Fredholm determinant using one or another powerful method developed for calculation of such determinants. The background field method \cite{DeWGAQFT.11,BuchOdinShap.11,WeinB2} simply relates the vacuum effective action on a given background to the effective action on a trivial background (the Minkowski background, in the case of gravity). As a rule, the background field method is a more powerful and consistent procedure to find the effective action than the immediate summation of the Feynman diagrams. It sums an infinite number of diagrams, and the result, in most cases, turns out to be essentially non-perturbative, i.e., non-analytic in the coupling constant.

In this paper, we investigate the dependence of the one-loop effective action calculated by using the background field method on the choice of a smooth globally defined foliation of the spacetime on the space and time, viz., on the choice of the time variable that enters the definition of the quantum Hamiltonian. Any systematic study of this question in the relativistic quantum field theory (QFT) framework with an accurate construction of the representation of the algebra of observables in the Fock space seems to be absent in the literature. We may only distinguish the papers by Kucha\v{r} \cite{KucharI,KucharII}, where the self-consistency of quantum gauge algebra of a parametrized massless scalar field on a two-dimensional Minkowski background was investigated. In particular, the quantum gravitational anomaly was found in these papers. It explicitly depends on the Killing vector defining the Hamiltonian. This anomaly is canceled out by the appropriate counterterm (by the redefinition of the quantum gauge algebra) which explicitly depends on the Killing vector. The issue of locality of this counterterm was not investigated in \cite{KucharI,KucharII}. Other gravitational anomalies were found in \cite{AlvaWitt.12,KamenLyakh}.

A possible reason for the absence of studies in this direction is that the formal expression for the one-loop correction to the effective action (see, e.g., \cite{BirDav.11}),
\begin{equation}\label{trace_gen}
    \Ga^{(1)}\sim\int_0^\infty\frac{d\tau}{\tau}\Sp e^{-i\tau (H-i0)},
\end{equation}
following from the so-called Schwinger variational principle \cite{DeWGAQFT.11}, looks as generally covariant. Here $H$ is a wave operator of a hyperbolic type, for instance, the Klein-Gordon operator. One may naively think that if one transforms the metric entering the wave operator according to the tensor law and the measure in the trace definition then the expression \eqref{trace_gen} does not change. However, the problem is that \eqref{trace_gen} is ill-defined since
\begin{enumerate}
  \item The operator $H$ realized on smooth square-integrable functions on the spacetime (it is that Hilbert space which is usually implied in \eqref{trace_gen}) possesses the spectrum unbounded from above and below;
  \item The operator under the trace sign in \eqref{trace_gen} is not trace-class.
\end{enumerate}
The first property says that the trace \eqref{trace_gen} as a function of $\tau$ can be defined only for pure imaginary $\tau$ and only in the sense of distributions of $\tau$. This is a rather serious obstacle for a correct evaluation of the effective action with the help of \eqref{trace_gen} since the trace in \eqref{trace_gen} is integrated with a non-compact function of $\tau$. Notice that the mathematically rigorous theory for the evaluation of traces of the operators of the form \eqref{trace_gen} is elaborated only for self-adjoint operators $H$ with the spectrum bounded from above or below \cite{Gilkey.6,Gilkey2.11,BerezMSQ1.4}. However, a severer problem comes from the second property which implies that the value of the trace in \eqref{trace_gen} depends on the choice of the basis (a sequence of the complete set of vectors in the Hilbert space), where the trace is evaluated, and for some bases it is not defined. This problem was cursorily mentioned in \cite{DeWGAQFT.11}, Vol. 2, p. 555, and with more details in \cite{GFSh.3}, Sect. 6.2, 6.3. Notice that this problem is not directly related to the ultraviolet divergencies in QFT that arise from the integration over $\tau$ in \eqref{trace_gen}. It is also clear that this problem cannot be cured by addition of any complex constant to $H$. The standard means to resolve this problem is to select the ``correct'' basis where the trace \eqref{trace_gen} should be evaluated, i.e., the problem of choice of the set of mode functions appears. The existence of such an ambiguity in the definition of the effective action is to be expected since, in order to construct relativistic QFT, one needs, apart from the wave operator defining the Heisenberg equations in the one-loop approximation, to specify the representation of the field operators in the Fock space, i.e., to define the splitting of the operators onto positive- and negative-frequency parts at every instant of time (see, e.g., \cite{DeWpaper.9}). The explicit realization of this construction will be given in Sect. \ref{Hamilt_Diag}. Let us stress that this problem is inherent to relativistic QFT. This issue is absent in non-relativistic many particle quantum problems where there is no need to split the operators onto positive- and negative-frequency part. This problem is also absent for models with a finite number of degrees of freedom.

In order that \eqref{trace_gen} be invariant under the action of diffeomorphisms, it is also necessary, in making the generally coordinate transformations, to transform the structure (as we shall see, the time-like vector field) distinguishing, with the aid of the Hamiltonian, the ``correct'' basis in the Hilbert space. Therefore, it is to be expected that the one-loop effective action depends not only on the background dynamical fields but also on the non-dynamical structure specifying the splitting of the operators onto positive- and negative-frequency parts. It is the main aim of this paper to study of this dependence. In contrast to \cite{KucharI,KucharII}, we shall consider the massive scalar field in the large mass limit. This will allow us, in particular, to discard the contributions of the gravitons and ghosts to the effective action since, in the one-loop approximation, these contributions are independent from the mass of a scalar field and cannot cancel the scalar field contribution.

In Sect. \ref{Hamilt_Diag}, the general procedure of diagonalization of the positive-definite quantum-field Hamiltonian for bosonic field is developed. It is shown that the Hamiltonian specifies a natural splitting of the field operators onto positive- and negative-frequency parts. In Sect. \ref{Scal_Field}, the general procedure of the Hamiltonian diagonalization is applied to the massive scalar field on a globally hyperbolic metric background. The unambiguous expression for the one-loop correction to the effective action induced by this field is given. We also prove the theorem that the imaginary part of the one-loop effective action is invariant under the metric diffeomorphisms keeping intact the initial and final Cauchy surfaces and the spatial infinity. Starting from Sect. \ref{Chang_Coord}, we consider QFT of a massive scalar field on the two-dimensional Minkowski background. In Sect. \ref{Chang_Coord}, a smooth change of coordinates identical at the spatial and temporal infinities is constructed in the Minkowski spacetime with the metric $\eta_{\mu\nu}=diag(1,-1)$. This change of coordinates is such constructed that we can explicitly, in a certain approximation, find the one-loop effective action induced by the scalar field in the new system of coordinates. As a result, the evaluation of the effective action is reduced to finding the transition amplitude for the one-dimensional Schr\"{o}dinger equation (Sect. \ref{Schr_Equation}) with the potential from a certain class. In Sect. \ref{Schr_Equation}, we also obtain the explicit expressions for the divergent and finite parts of the effective action. In particular, it is shown that the contribution standing at $m^{-2}$ in the large mass expansion of the effective action is not zero. In Sect. \ref{Example}, we prove that the real part of effective action contains exponentially suppressed non-analytic in $m^{-2}$ contributions for a certain family of the potentials from the class mentioned above. It turns out that there are no exactly solvable family of potentials in the class of potentials considered (for our purposes, it is necessary to find a family of potentials with certain properties, see Sect. \ref{Schr_Equation}, rather than one potential). Therefore, we consider a family of exactly solvable potentials that approximates with a good accuracy a certain family of potentials from the required class and prove that the corrections to the one-loop effective action caused by the potential difference are small and cannot cancel the non-analytic in $m^{-2}$ terms. In conclusion, we formulate the main results. In App. \ref{Second_Quant}, we provide the necessary formulas from \cite{BerezMSQ1.4} and derive the well-defined expression for the one-loop correction to the effective action. In App. \ref{Dens_Stat}, the formal derivation of the relation between the $S$-matrix and the density of states of a quantum system is given. In App. \ref{Adiabatic_Invariant}, some analytic properties and asymptotics of the adiabatic invariant as a function of the complex energy are studied. In App. \ref{Reg_Renorm}, we collect some necessary facts from renormalization theory. In order to understand the conclusions made in the corresponding section, the reader is advised to become acquainted with the background field method (see, e.g., \cite{DeWGAQFT.11,BuchOdinShap.11,WeinB2}).

Briefly, the main result of this paper can be formulated as follows. It is well-known that the one-loop correction to the effective action induced by a massive scalar field on the two-dimensional Minkowski spacetime with the metric $\eta_{\mu\nu}$ is given by the Coleman-Weinberg potential \cite{ColWein.9}. If one makes a smooth change of coordinates, constructs QFT in this new system of coordinates, and finds the one-loop correction then one may expect, assuming the invariance of the effective action under the metric diffeomorphisms, that one again obtains the Coleman-Weinberg potential with the transformed metric. However, as we shall see, this is not the case. The additional terms both analytic and non-analytic in $m^{-2}$ arise in the effective action apart from the Coleman-Weinberg potential. The non-analytic terms cannot be canceled out by the counterterms.

\section{Diagonalization of the Hamiltonian}\label{Hamilt_Diag}

Consider in detail the Hamiltonian diagonalization procedure in relativistic QFT. The procedure of the Hamiltonian diagonalization expounded below is an infinite dimensional analogue of the procedure given in \cite{BabBuld,MaslovCWKB,BBTZ,BBT.1}. For the Friedmann-Lema\^{\i}tre metric, it is elaborated in \cite{GriMaMos.11}.

Let the Hamiltonian of a quantum-field system in the Schr\"{o}dinger representation be
\begin{equation}\label{Hamilt_quadr}
    \hat{H}(t)=\frac12 \hat{Z}^AH_{AB}(t)\hat{Z}^B,\qquad \hat{Z}^A=\left[
                                           \begin{array}{c}
                                             \hat{\phi}(\spx) \\
                                             \hat{\pi}(\spx) \\
                                           \end{array}
                                         \right],\qquad
    [\hat{Z}^A,\hat{Z}^B]=iJ^{AB}=\left[
                        \begin{array}{cc}
                          0 & i \\
                          -i & 0 \\
                        \end{array}
                      \right]\de(\spx-\spy).
\end{equation}
It is supposed that
\begin{equation}
    \bar{H}_{AB}=H_{BA}=H_{AB},
\end{equation}
and $H_{AB}$ defines the positive-definite quadratic form. Henceforth, the bar over the expression means a complex conjugation, unless otherwise stated.

Let us pose the eigenvalue problem for the non-singular self-conjugate operator $-iJ_{AB}=(iJ^{AB})^{-1}$ with respect to the quadratic form $H_{AB}$:
\begin{equation}
    -iJ_{AB}\ups^B_\al(t)=\mu_\al(t) H_{AB}(t)\ups^B_\al(t),
\end{equation}
where $\ups_\al^A$ obey certain boundary conditions dictated by the physical problem statement. Inasmuch as $J_{AB}$ is non-singular, $\mu_\al\neq0$. It follows from the reality of $H_{AB}$ that the eigenvalues comes into pairs $(\mu_\al,-\mu_\al)$, moreover if the vector $\ups_\al$ corresponds to the eigenvalue $\mu_\al$ then the vector $\bar{\ups}_\al^A$ corresponds to the eigenvalue $-\mu_\al$. Therefore, further we suppose that $\mu_\al>0$ and
\begin{equation}\label{eigen_prblm}
    -iJ_{AB}\ups^B_\al(t)=\mu_\al(t) H_{AB}(t)\ups^B_\al(t),\qquad iJ_{AB}\bar{\ups}^B_\al(t)=\mu_\al(t) H_{AB}(t)\bar{\ups}^B_\al(t).
\end{equation}
The orthogonality and completeness relations are equivalent to
\begin{equation}\label{orth_complt}
    \{\ups_\al,\ups_\be\}=\{\bar{\ups}_\al,\bar{\ups}_\be\}=0,\quad\{\ups_\al,\bar{\ups}_\be\}=-i\de_{\al\be},\qquad iJ^{AB}=\sum_\al(\ups_\al^A\bar{\ups}_\al^B-\bar{\ups}_\al^A\ups_\al^B),
\end{equation}
where $\{\ups,w\}:=J_{AB}\ups^A w^B$, and the normalization of eigenfunctions is chosen to be
\begin{equation}
    \bar{\ups}^A_\al H_{AB}\ups_\al^B=\ups^A_\al H_{AB}\bar{\ups}_\al^B=\mu_\al^{-1}.
\end{equation}
In other words, the vectors ($\ups_\al$,$\bar{\ups}_\al$) constitute the symplectic basis.

The operators,
\begin{equation}\label{creaannh_oper}
    \hat{a}_\al(t):=\{\bar{\ups}_\al(t),\hat{Z}\},\qquad \hat{a}^\dag_\al(t):=\{\ups_\al(t),\hat{Z}\},
\end{equation}
obey the creation-annihilation operators commutation relations
\begin{equation}
    [\hat{a}_\al,\hat{a}^\dag_\be]=\de_{\al\be},\qquad [\hat{a}_\al,\hat{a}_\be]=0=[\hat{a}^\dag_\al,\hat{a}^\dag_\be].
\end{equation}
Also the following representation holds
\begin{equation}\label{ZinAAd}
    \hat{Z}^A=-i\sum_\al(\ups_\al^A\hat{a}_\al-\bar{\ups}_\al^A\hat{a}^\dag_\al).
\end{equation}
As we see, the Hamiltonian \eqref{Hamilt_quadr} specifies a natural splitting of the field operators onto positive- and negative-frequency parts.

Putting $\omega_\al(t):=\mu_\al^{-1}(t)$ and substituting \eqref{ZinAAd} into \eqref{Hamilt_quadr}, we obtain
\begin{equation}
    \hat{H}=\frac12\sum_\al\omega_\al(t)[\hat{a}^\dag_\al(t)\hat{a}_\al(t)+\hat{a}_\al(t) \hat{a}^\dag_\al(t)].
\end{equation}
This operator is not defined in the Fock space with the vacuum annihilated by $\hat{a}_\al(t)$ provided $\sum_\al\omega_\al(t)$ diverges, which usually occurs. Let us introduce the regularized Hamilton operator
\begin{equation}\label{Hamilt_reg}
    \hat{H}_\La(t)=\frac12\sum_\al^\La\omega_\al(t)[\hat{a}^\dag_\al(t)\hat{a}_\al(t)+\hat{a}_\al(t) \hat{a}^\dag_\al(t)]=\sum_\al^\La \omega_\al(t)\hat{a}^\dag_\al(t)\hat{a}_\al(t)+\frac12\sum_\al^\La\omega_\al(t),
\end{equation}
where the sum over $\al$ is carried over those eigenvalues that correspond to the energies $\omega_\al(t)$ less than $\La$. It is clear that
\begin{equation}
    \hat{H}_\La\underset{\La\rightarrow+\infty}{\rightarrow}\hat{H}.
\end{equation}
The divergencies arising in the regularization removal limit must be canceled out by addition of the appropriate counterterms to the initial classical action.

Now we pass to the basis in the Fock space which is constituted by the eigenvectors of the instantaneous Hamiltonian $\hat{H}_\La(t)$. To this end, we introduce the unitary operator $\hat{W}_{t,t_{in}}$ such that (cf. \cite{Kato})
\begin{equation}
\begin{gathered}
    \hat{a}_\al(t_{in})=\hat{W}_{t,t_{in}}\hat{a}_\al(t)\hat{W}_{t_{in},t},\quad \hat{a}^\dag_\al(t_{in})=\hat{W}_{t,t_{in}}\hat{a}^\dag_\al(t)\hat{W}_{t_{in},t},\quad
    |out\ran:=\hat{W}_{t_{out},t_{in}}|out,t_{out}\ran,\quad|in\ran:=|in,t_{in}\ran,\\
    \lan out,t_{out}|\hat{U}_{t_{out},t_{in}}|in,t_{in}\ran=\lan out|\hat{S}_{t_{out},t_{in}}|in\ran,\qquad \hat{S}_{t,t_{in}}:=\hat{W}_{t,t_{in}}\hat{U}_{t,t_{in}}.
\end{gathered}
\end{equation}
where $\hat{U}_{t,t_{in}}$ is the evolution operator generated by the Hamiltonian $\hat{H}_\La(t)$. The necessary and sufficient conditions for the existence of the operator $\hat{W}_{t,t_{in}}$ are given in the theorem \ref{canon_unit_thm} in App. \ref{Second_Quant}. The operator $\hat{S}_{t,t_{in}}$ obeys the equation
\begin{equation}\label{Smatr_energy}
    i\partial_t\hat{S}_{t,t_{in}}=\Big\{i\partial_t\hat{W}_{t,t_{in}}\hat{W}_{t_{in},t}+ \frac12\sum_\al^\La\omega_\al(t)[\hat{a}^\dag_\al(t_{in})\hat{a}_\al(t_{in})+\hat{a}_\al(t_{in}) \hat{a}^\dag_\al(t_{in})]\Big\}\hat{S}_{t,t_{in}}.
\end{equation}
In order to find the explicit expression for the first term in the curly brackets, we employ the fact that
\begin{equation}
    i\dot{\hat{a}}_\al(t)=[\hat{a}_\al(t),\hat{W}_{t_{in},t}i\partial_t\hat{W}_{t,t_{in}}],\qquad i\dot{\hat{a}}^\dag_\al(t)=[\hat{a}^\dag_\al(t),\hat{W}_{t_{in},t}i\partial_t\hat{W}_{t,t_{in}}].
\end{equation}
Then
\begin{equation}
    i\partial_t\hat{W}_{t,t_{in}}=:\hat{W}_{t,t_{in}}\hat{K}(t,\hat{a}(t),\hat{a}^\dag(t))=\hat{W}_{t,t_{in}}\hat{K}(t,\hat{a}(t),\hat{a}^\dag(t))\hat{W}_{t_{in},t}\hat{W}_{t,t_{in}} =\hat{K}(t,\hat{a}(t_{in}),\hat{a}^\dag(t_{in}))\hat{W}_{t,t_{in}}.
\end{equation}
For the operator
\begin{equation}
    \hat{K}(t,\hat{a}(t),\hat{a}^\dag(t))=\frac12\big[2\hat{a}^\dag(t) C(t)\hat{a}(t)+a(t) \bar{A}(t) \hat{a}(t)+\hat{a}^\dag(t) A(t) \hat{a}^\dag\big]
\end{equation}
we deduce that
\begin{equation}
    i\dot{\hat{a}}(t)=[\hat{a}(t),\hat{K}(t)]=C(t)\hat{a}(t)+A(t)\hat{a}^\dag(t),\qquad i\dot{\hat{a}}^\dag(t)=[\hat{a}^\dag(t),\hat{K}(t)]=-\bar{A}(t)\hat{a}(t)-\bar{C}(t)\hat{a}^\dag(t).
\end{equation}
The solution to these equations takes the form
\begin{equation}
    \left[
      \begin{array}{c}
        \hat{a}(t) \\
        \hat{a}^\dag(t) \\
      \end{array}
    \right]=
    \left[
      \begin{array}{cc}
        F(t) & G(t) \\
        \bar{G}(t) & \bar{F}(t) \\
      \end{array}
    \right]
    \left[
      \begin{array}{c}
        \hat{a}(t_{in}) \\
        \hat{a}^\dag(t_{in}) \\
      \end{array}
    \right]=:\mathcal{A}(t)
    \left[
      \begin{array}{c}
        \hat{a}(t_{in}) \\
        \hat{a}^\dag(t_{in}) \\
      \end{array}
    \right].
\end{equation}
On the other hand, with the aid of \eqref{orth_complt}, it follows from \eqref{ZinAAd} taken at the instants of time $t_{in}$ and $t$ that
\begin{equation}
    F_{\al\be}=-i\{\bar{\ups}_\al(t),\ups_\be(t_{in})\},\qquad G_{\al\be}=i\{\bar{\ups}_\al(t),\bar{\ups}_\be(t_{in})\}.
\end{equation}
Consequently,
\begin{equation}
    \left[
      \begin{array}{cc}
        C(t) & A(t) \\
        -\bar{A}(t) & -\bar{C}(t) \\
      \end{array}
    \right]=i\dot{\mathcal{A}}(t)\mathcal{A}^{-1}(t)=
    \left[
      \begin{array}{cc}
        \{\dot{\bar{\ups}}_\al(t),\ups_\be(t)\} & -\{\dot{\bar{\ups}}_\al(t),\bar{\ups}_\be(t)\} \\
        \{\dot{\ups}_\al(t),\ups_\be(t)\} & -\{\dot{\ups}_\al(t),\bar{\ups}_\be(t)\} \\
      \end{array}
    \right],
\end{equation}
where the completeness relation \eqref{orth_complt} of the mode functions was used. As a result, we can write \eqref{Smatr_energy} as
\begin{multline}\label{Smatr_energy_1}
    i\partial_t\hat{S}_{t,t_{in}}=\Big\{\frac12\sum_\al^\La\omega_\al(t)[\hat{a}^\dag_\al(t_{in})\hat{a}_\al(t_{in})+\hat{a}_\al(t_{in}) \hat{a}^\dag_\al(t_{in})]\\ +\hat{a}^\dag_\al(t_{in})\{\dot{\bar{\ups}}_\al,\ups_\be\}\hat{a}_\be(t_{in}) -\frac12\hat{a}_\al(t_{in})\{\dot{\ups}_\al,\ups_\be\}\hat{a}_\be(t_{in}) -\frac12\hat{a}^\dag_\al(t_{in})\{\dot{\bar{\ups}}_\al,\bar{\ups}_\be\}\hat{a}_\be^\dag(t_{in}) \Big\}\hat{S}_{t,t_{in}}.
\end{multline}

It is useful to express the skew-scalar products entering the last expression through the derivatives of the quadratic form $H_{AB}$ with respect to time. Differentiating \eqref{eigen_prblm} with respect to $t$ and contracting with the mode functions, we obtain
\begin{equation}
    \{\dot{\ups}_\al,\ups_\be\}=\frac{\ups_\al^A\dot{H}_{AB}\ups_\be^B}{i(\omega_\al+\omega_\be)}, \qquad\{\dot{\bar{\ups}}_\al,\ups_\be\}_{\al\neq\be}=\frac{\bar{\ups}_\al^A\dot{H}_{AB}\ups_\be^B}{i(\omega_\be-\omega_\al)}.
\end{equation}
For $\al=\be$, we have the relations
\begin{equation}
    \ups_\al^A H_{AB}\bar{\ups}_\al^B=\omega_\al,\qquad \ups_\al^A \dot{H}_{AB}\bar{\ups}_\al^B=\dot{\omega}_\al.
\end{equation}
Whence, differentiating the first equality, we come to
\begin{equation}
    \im\{\dot{\bar{\ups}}_\al,\ups_\al\}=0,
\end{equation}
for $\omega_\al\neq0$. The fact that $\re\{\dot{\bar{\ups}}_\al,\ups_\al\}$ is not expressed through the derivatives of the quadratic form $H_{AB}$ is a consequence of invariance of the system of equations \eqref{eigen_prblm}, \eqref{orth_complt} under the transformations
\begin{equation}
    \ups_\al^A(t)\rightarrow e^{i\vf_\al(t)}\ups_\al^A(t).
\end{equation}
We show in App. \ref{Second_Quant} that the one-loop $in$-$out$ effective action is invariant under these transformations, and so $\re\{\dot{\bar{\ups}}_\al,\ups_\al\}$ can be taken to be equal to any smooth function of $t$ for every $\al$.

\section{Scalar field}\label{Scal_Field}

The action for a massive scalar field on the background with the metric $g_{\mu\nu}$ has the standard form
\begin{equation}
    S[\phi]=\frac12\int d^Dx\sqrt{|g|}(\partial_\mu\phi g^{\mu\nu}\partial_\nu\phi-m^2\phi^2).
\end{equation}
Let us introduce the momentum
\begin{equation}\label{Legandre_trans}
    \pi:=\frac{\partial \mathcal{L}}{\partial\dot{\phi}}=\sqrt{|g|}(g^{00}\dot{\phi}+g^{0i}\partial_i\phi),
\end{equation}
where $\dot{\phi}=\partial_t\phi$, and $\mathcal{L}$ is the Lagrangian density. Making the Legendre transform, we obtain the Hamiltonian density
\begin{multline}\label{Hamilt_scal}
    \mathcal{H}=\pi\dot{\phi}-\mathcal{L}=\frac12\Big[\frac{\pi^2}{g^{00}\sqrt{|g|}} -2\frac{\pi g^{0i}\partial_i\phi}{g^{00}} -\sqrt{|g|}(\tilde{g}^{ij}\partial_i\phi\partial_j\phi-m^2\phi^2) \Big]=\\
    =\frac12\Big[\frac{(\pi-\sqrt{|g|}g^{0i}\partial_i\phi)^2}{g^{00}\sqrt{|g|}}-\sqrt{|g|}(g^{ij}\partial_i\phi\partial_j\phi-m^2\phi^2 )\Big],
\end{multline}
where $\tilde{g}^{ij}=g^{ij}-g^{0i}g^{0j}/g^{00}=(g_{ij})^{-1}$. Obviously, the definition of the Hamiltonian density depends explicitly on a choice of the time variable $t$. The last expression for $\mathcal{H}$ shows that if $g^{00}>0$ and $g^{ij}$ is negative-definite then the Hamiltonian defines the positive-definite quadratic form $H_{AB}$ in the notation of Sect. \ref{Hamilt_Diag}. These conditions are sufficient but not necessary for a positive definiteness of $H_{AB}$. Notice that if $g_{\mu\nu}$ is the Minkowski metric written in the curvilinear coordinates then the Hamiltonian constructed with the help of the Hamilton density \eqref{Hamilt_scal} does not coincide with the generator of translations along time of the Poincar\'{e} algebra.

The quadratic form $H_{AB}$ is written as
\begin{equation}
    H_{AB}=\left[
             \begin{array}{cc}
               \partial_i\sqrt{|g|}\tilde{g}^{ij}\partial_j+\sqrt{|g|}m^2 & \partial_i\frac{g^{i0}}{g^{00}} \\
               -\frac{g^{0i}}{g^{00}}\partial_i & \frac{1}{g^{00}\sqrt{|g|}} \\
             \end{array}
           \right].
\end{equation}
The eigenvalue problem \eqref{eigen_prblm} takes the form
\begin{equation}
    \left[
       \begin{array}{c}
         \partial_i(\sqrt{|g|}\tilde{g}^{ij}\partial_ju_\al)+\sqrt{|g|}m^2u_\al+\partial_i(\frac{g^{i0}}{g^{00}}w_\al) \\
         -\frac{g^{0i}}{g^{00}}\partial_iu_\al+\frac{w_\al}{g^{00}\sqrt{|g|}} \\
       \end{array}
     \right]=i\omega_\al\left[
                          \begin{array}{c}
                            w_\al \\
                            -u_\al \\
                          \end{array}
                        \right],\qquad\ups^A_\al(t)=\left[
                                                   \begin{array}{c}
                                                     u_\al(t) \\
                                                     w_\al(t) \\
                                                   \end{array}
                                                 \right].
\end{equation}
Combining these expressions, we come to the equations
\begin{equation}\label{KG_stat}
    \Big[(p_i+\omega_\al g_i)\sqrt{|g|}g^{ij}(p_j+\omega_\al g_j)+\sqrt{|g|}\big(\frac{\omega^2_\al}{g_{00}}-m^2\big)\Big]u_\al=0,\qquad w_\al=-i\sqrt{|g|}\Big[g_i g^{ij}(p_j+\omega_\al g_j)+\frac{\omega_\al}{g_{00}}\Big]u_\al,
\end{equation}
where $p_i=-i\partial_i$ and $g_i:=g_{0i}/g_{00}$. In fact, we need to solve the first equation, which is nothing but the stationary Klein-Gordon equation (see, e.g.,  \cite{psfss.10,KalKaz.11,KalKaz2.12}), supplied with the proper boundary conditions. Notice that the metric, the eigenfunctions $u_\al$, $w_\al$, and the eigenvalues $\omega_\al$ depend on the time $t$.

Our goal is to investigate the dependence of the one-loop $in$-$out$ effective action on the choice of the time variable in the background field method framework. This one-loop correction is expressed through the vacuum-vacuum transition amplitude (see, e.g., \cite{BirDav.11,DeWGAQFT.11,GFSh.3,WeinB2})
\begin{equation}
    i\Ga^{(1)}_{in-out}=\ln\lan out,t_{out}|\hat{U}_{t_{out},t_{in}}|in,t_{in}\ran,
\end{equation}
where $|in,t_{in}\ran$ and $|out,t_{out}\ran$ are the vacuum states which we take as the ground eigenvectors of the Hamiltonians $\hat{H}(t_{in})$ and $\hat{H}(t_{out})$, respectively. In particular, making use of the representation \eqref{Smatr_energy_1} and formulas \eqref{vac_vac}, \eqref{vac_vac_expl}, we deduce
\begin{equation}\label{inout_eff_act}
    i\Ga^{(1)}_{in-out}=\ln\lan out|\hat{S}_{t_{out},t_{in}}|in\ran=-\frac{1}{2}\ln\det[\bar{\Phi}(t_{out})V(t_{out})]-\frac{i}2\int_{t_{in}}^{t_{out}}d\tau\sum_\al^\La\omega_\al(\tau),
\end{equation}
where the operators $\Phi(t)$, $V(t)$, $C(t)$, and $A(t)$ are given in \eqref{unit_canon}, \eqref{V_oper}, and \eqref{CA_energy}. The first term in \eqref{inout_eff_act} can be calculated by the use of the perturbation theory presented at the end of App. \ref{Second_Quant}. This term is absent when $A(t)=0$. The second term is the contribution of the instantaneous energy of zero-point fluctuations. In the case when the background is stationary and the time variable is chosen such that the Killing vector has the form $\partial_t$, this term describes the Casimir effect. Evaluating this term on the appropriate background, one can find, for example, all the one-loop corrections in quantum electrodynamics \cite{DeWGAQFT.11}. In order to calculate this term, one can employ the high-temperature expansion technique elaborated in \cite{Fursaev1.10,Fursaev2.10,KalKaz.11,KalKaz2.12} for arbitrary stationary backgrounds. Since $\omega_\al(t)$ are found from the stationary Klein-Gordon equation \eqref{KG_stat}, which has the same form as [(4.4), \cite{KalKaz2.12}], the formulas obtained in \cite{KalKaz2.12} remain intact with the only difference that now the metric and the whole expression depend on $t$. The theorem \ref{evol_symb_thm} gives the sufficient conditions for \eqref{inout_eff_act} to exist and be well-defined at the fixed $\La$. Unfortunately, we have not succeeded in joining in a consistent way (i.e., up to the term independent from the background field) the two terms in \eqref{inout_eff_act} into one as in the formal expression \eqref{trace_gen}.

Now we pass into the Heisenberg representation
\begin{equation}
\begin{gathered}
    \hat{a}(in)=\hat{U}_{0,t_{in}}\hat{a}(t_{in})\hat{U}_{t_{in},0},\qquad \hat{a}^\dag(in)=\hat{U}_{0,t_{in}}\hat{a}^\dag(t_{in})\hat{U}_{t_{in},0},\\
    \hat{a}(out)=\hat{U}_{0,t_{out}}\hat{a}(t_{out})\hat{U}_{t_{out},0},\qquad \hat{a}^\dag(out)=\hat{U}_{0,t_{out}}\hat{a}^\dag(t_{out})\hat{U}_{t_{out},0},
\end{gathered}
\end{equation}
where $\hat{a}(t_{in})$ and $\hat{a}(t_{out})$ are the annihilation operators constructed according to \eqref{creaannh_oper} for $t=t_{in}$ and $t=t_{out}$, respectively. It is clear that the states
\begin{equation}
    |\overline{in}\ran:=\hat{U}_{0,t_{in}}|in,t_{in}\ran,\qquad |\overline{out}\ran:=\hat{U}_{0,t_{out}}|out,t_{out}\ran,
\end{equation}
are the vacuum ones for the operators $\hat{a}(in)$ and $\hat{a}(out)$, respectively, and
\begin{equation}
    \lan out,t_{out}|\hat{U}_{t_{out},t_{in}}|in,t_{in}\ran=\lan \overline{out}|\overline{in}\ran.
\end{equation}
Let
\begin{equation}\label{Heis_eqs}
    \hat{Z}^A(t):=\hat{U}_{0,t}\hat{Z}^A\hat{U}_{t,0},\qquad i\dot{\hat{Z}}^A(t)=[\hat{Z}^A(t),\hat{H}(t)],
\end{equation}
where $\hat{H}(t)$ is written in the Heisenberg representation. From \eqref{ZinAAd}, we obtain
\begin{equation}\label{Zoutin}
    \hat{Z}^A(t_{out})=-i\sum_\al\big[\ups_\al^A(t_{out})\hat{a}_\al(out)-\bar{\ups}_\al^A(t_{out})\hat{a}^\dag_\al(out)\big],\qquad \hat{Z}^A(t_{in})=-i\sum_\al\big[\ups_\al^A(t_{in})\hat{a}_\al(in)-\bar{\ups}_\al^A(t_{in})\hat{a}^\dag_\al(in)\big].
\end{equation}
On the other hand, introducing the commutator Green function
\begin{equation}
    \tilde{G}^A_{\ B}(t,t'):=[\hat{Z}^A(t),\hat{Z}^B(t')],
\end{equation}
and using the commutation relations \eqref{Hamilt_quadr}, we can write
\begin{equation}\label{Zout}
    \hat{Z}^A(t_{out})=-i\tilde{G}^A_{\ B}(t_{out},t_{in}) \hat{Z}^B(t_{in}).
\end{equation}
Notice that $\tilde{G}^A_{\ B}(t,t')$ is a $c$-number for quadratic Hamiltonians \eqref{Hamilt_quadr}, i.e., in the one-loop approximation. It follows from the relations \eqref{creaannh_oper}, \eqref{Zoutin}, \eqref{Zout} that
\begin{equation}\label{canon_trans_expl}
\begin{split}
    \hat{a}(out)&=-\bar{\ups}_\al^A(t_{out})\tilde{G}_{AB}(t_{out},t_{in})\ups^B_\be(t_{in})\hat{a}_\be(in) +\bar{\ups}_\al^A(t_{out})\tilde{G}_{AB}(t_{out},t_{in})\bar{\ups}^B_\be(t_{in})\hat{a}^\dag_\be(in),\\
    \hat{a}^\dag(out)&=-\ups_\al^A(t_{out})\tilde{G}_{AB}(t_{out},t_{in})\ups^B_\be(t_{in})\hat{a}_\be(in) +\ups_\al^A(t_{out})\tilde{G}_{AB}(t_{out},t_{in})\bar{\ups}^B_\be(t_{in})\hat{a}_\be(in),
\end{split}
\end{equation}
where $\tilde{G}_{AB}=J_{AC}\tilde{G}^C_{\ B}$, i.e., we have a linear canonical transformation of the form \eqref{canon_trans}.

In the regularization removal limit $\La\rightarrow+\infty$, the commutator Green function can be found directly from the Heisenberg equations \eqref{Heis_eqs}. Equations \eqref{Heis_eqs} allow for the regularization removal limit and in this limit give the Klein-Gordon equation for the $\hat{\phi}$ component of $\hat{Z}^A$:
\begin{equation}\label{KleinGord_eqn}
    (\nabla^2+m^2)\hat{\phi}=0.
\end{equation}
Let the spacetime be globally hyperbolic. Then equation \eqref{KleinGord_eqn} possesses the unique retarded Green function satisfying the equation
\begin{equation}\label{retGrenfun_def}
    (\nabla^2_x+m^2)G^-(x,y)=-\frac{\de(x-y)}{\sqrt{|g(x)|}},\qquad G^-(x,y)=0\;\text{for}\; x^0<y^0.
\end{equation}
Integrating two times by parts the left hand side of the identity
\begin{equation}
    \int_\Omega dy\sqrt{|g(y)|}G^-(x,y)(\nabla^2+m^2)\hat{\phi}=0,
\end{equation}
and using \eqref{retGrenfun_def}, we arrive at
\begin{equation}
    \hat{\phi}(x)=\int_\Si d\Si^\mu\big[G^-(x,y)\partial_\mu\hat{\phi}(y)-\partial^y_\mu G^-(x,y)\hat{\phi}(y)\big],
\end{equation}
where $\Si=\partial\Omega$ is the boundary of the region $\Omega$ and $d\Si^\mu$ is the area element of the surface. Setting $\Omega$ be a cylinder with the bases $t=t_{in}$, $t=t_{out}$ and taking into account the property \eqref{retGrenfun_def} of the retarded Green function, we have
\begin{equation}\label{Cauchy_sol}
    \hat{\phi}(x)=\int d\spy\big[-\sqrt{|g(y)|}g^{0\mu}(y)\partial_\mu^yG^-(x,y)\hat{\phi}(y)+ G^-(x,y)\hat{\pi}(y)\big],\quad y^0=t_{in}.
\end{equation}
As a result, we obtain
\begin{equation}
    \tilde{G}_{AB}(t,t')=
    \left[
      \begin{array}{cc}
        \sqrt{|g(x)|}g^{0\mu}(x)\sqrt{|g(y)|}g^{0\nu}(y)\partial^{xy}_{\mu\nu} & -\sqrt{|g(x)|}g^{0\mu}(x)\partial_\mu^x \\
        -\sqrt{|g(y)|}g^{0\mu}(y)\partial_\mu^y & 1 \\
      \end{array}
    \right]iG^-(x,y),\qquad x^0=t>y^0=t',
\end{equation}
for the commutator Green function.

Now we can formulate the statement

\begin{thm}
  The imaginary part of the one-loop effective action induced by a scalar field on a globally hyperbolic spacetime allows for the regularization removal limit and, in this limit, does not depend on the choice of the smooth splitting of the spacetime onto the space and time (i.e., on the choice of the time variable) for $t\in(t_{in},t_{out})$ such that the corresponding Hamiltonian generates the unitary evolution. The imaginary part of the one-loop effective action is determined solely by the mass $m$, the spacetime metric, and by the choice of the time variable at the instants of time $t=t_{in}$ and $t=t_{out}$.
\end{thm}
\begin{proof}
  The imaginary part of the one-loop effective action has the form
  \begin{equation}
    \im \Ga^{(1)}_{in-out}=-\ln|\lan out,t_{out}|\hat{U}_{t_{out},t_{in}}|in,t_{in}\ran|=-\ln|\lan \overline{out}|\overline{in}\ran|.
  \end{equation}
  According to the theorem \ref{canon_unit_thm}, the modulus of the vacuum-vacuum transition amplitude is uniquely determined by the operators $\Phi$ and $\Psi$ relating the sets of the creation-annihilation operators. In the case at hand, these operators are given in \eqref{canon_trans_expl} and expressed, in the regularization removal limit, in terms of the mode functions, diagonalizing the Hamiltonian at the instants of time $t=t_{in}$ and $t=t_{out}$, and the retarded Green function of the Klein-Gordon operator.

  As seen from \eqref{Cauchy_sol}, the retarded Green function is the solution to the Cauchy problem with the initial data $\phi(t_{in},\spx)=0$, $\pi(t_{in},\spx)=\de(\spx-\spy)$. Then it is clear that $G^-(x,y)$ is invariant under the metric diffeomorphisms\footnote{Notice that we discriminate between diffeomorphisms (the active general coordinate transformations) and changes of coordinates (the passive general coordinate transformations). See Conclusion for details.} leaving intact the initial and final Cauchy surfaces and the spatial infinity (the surface $\Si$). Indeed, consider the two Klein-Gordon equations with metrics related by the diffeomorphism mentioned. Then we can make a change of variables in the second equation such that this equation looks as the first one in the region $\Omega$. The Cauchy data for the retarded Green function also coincide in this system of coordinates. Consequently, in such a system of coordinates, the retarded Green functions for the both Klein-Gordon equations coincide in the region $\Omega$. Since the change of variables used is identical in the neighbourhood of $\Si$ then $G^-(x,y)$ coincides for the both equations on $\Si$ in the initial system of coordinates.

  Thus the retarded Green function is independent from the smooth splitting onto the space and time. The requirements of the theorem \ref{evol_symb_thm} provide the sufficient conditions for existence of the unitary evolution for $t\in[t_{in},t_{out}]$.
\end{proof}

In particular, if the metric $g_{\mu\nu}$ entering the Hamiltonian \eqref{Hamilt_scal} is obtained from the Minkowski metric $\eta_{\mu\nu}$ with the aid of the diffeomorphism identical at $t\geq t_{out}$, $t\leq t_{in}$, and at the spatial infinity, and the vector field $\partial_t$ is time-like then the imaginary part of the one-loop effective action for QFT with the Hamiltonian \eqref{Hamilt_scal} vanishes in the regularization removal limit. The non-trivial dependence on mass that we are interested in can be contained only in the real part of the one-loop effective action.

\section{Change of the coordinates}\label{Chang_Coord}

Now we restrict our considerations to the case of a two-dimensional Minkowski spacetime. Let us construct the smooth change of coordinates in the whole Minkowski spacetime such that
\begin{enumerate}
  \item Equation \eqref{KG_stat} in this system of coordinates reduces to the one-dimensional Schr\"{o}dinger equation;
  \item The change of coordinates does not depend on $m^2$ and becomes identical at the spatial infinity and for $t\geq t_{out}$ and $t\leq t_{in}$;
  \item The metric in this system of coordinates weakly depends on time (the adiabaticity).
\end{enumerate}
The independence of the change of coordinates from $m^2$ is necessary for the metric components to be independent from $m^2$ in the new system of coordinates. Only in this case are the arguments, given in the Introduction, on impossibility to cancel the essentially singular in $m^{-2}$ contributions to the one-loop effective action by the ghosts and gravitons contributions applicable. The condition 3 is needed to neglect the contribution of the first term in \eqref{inout_eff_act} to the effective action. It will be rigorously formulated below.

Let
\begin{equation}\label{metrics}
    ds^2=dt'^2-dx'^2=g_{00}[dt^2+2fdtdx-(1-f^2)dx^2],\qquad g=\det g_{\mu\nu}=-g_{00}^2,
\end{equation}
where $g_{00}=g_{00}(t,x)$ and $f=f(t,x)$. Inasmuch as the change of coordinates considered is well defined in the whole Minkowski spacetime and it does not change, by definition, the vector field $\partial_t$ entering the Legendre transform \eqref{Legandre_trans}, i.e., we construct the Hamiltonian ``with respect to'' the vector field $\partial_t$ in the new system of coordinates rather than ``with respect to'' the time-like vector field originating from $\partial_{t'}$ as a result of the change of coordinates, we can regard this change of coordinates as the metric diffeomorphism. The inverse metric is
\begin{equation}
    g^{\mu\nu}=g_{00}^{-1}\left[
                            \begin{array}{cc}
                              1-f^2 & f \\
                              f & -1 \\
                            \end{array}
                          \right].
\end{equation}
Equations \eqref{KG_stat} can be cast into the form
\begin{equation}\label{KG_stat_2}
\begin{gathered}
    \big[p_x^2+m^2g_{00}(t,x)\big]\tilde{u}_\al(t,x)=\omega^2_\al(t)\tilde{u}_\al(t,x),\qquad \tilde{w}_\al=f\partial_x\tilde{u}_\al-i\omega_\al\tilde{u}_\al,\\
    u_\al(t,x)=:e^{iS_\al(t,x)}\tilde{u}_\al(t,x),
    \qquad w_\al(t,x)=:e^{iS_\al(t,x)}\tilde{w}_\al(t,x),\qquad S_\al(t,x):=-\omega_\al(t)\int_0^xdyf(t,y).
\end{gathered}
\end{equation}
Notice that the spectrum $\omega_\al$ does not depend on the function $f$. In order to fix uniquely the phases of the mode functions $u_\al(t,x)$, which is necessary to construct the unique operator $\hat{W}_{t,t_{in}}$ defined in Sect. \ref{Hamilt_Diag}, we demand that $\tilde{u}_\al(t,x)$ is real.

It is convenient to introduce the notation
\begin{equation}
    g_{00}(t,x)=:e^{2\psi(t,x)}>0.
\end{equation}
The condition $g_{00}>0$ guaranties that the metric \eqref{metrics} is physical \cite{LandLifshCTF.1}. We also assume that
\begin{equation}
    t_{out}=-t_{in}=:T/2.
\end{equation}
The function $\psi(t,x)$ obeys the conditions
\begin{equation}\label{bound_cond_psi}
\begin{gathered}
    \psi(t,x)\underset{|t|\geq T/2}{=}0,\qquad \psi(t,x)\underset{|x|\rightarrow\infty}{\rightarrow}0,\qquad \psi'(t,x)\underset{|x|\rightarrow\infty}{\rightarrow}0,\\
    \Big|\int_{-\infty}^\infty dx\dot{\psi}(t,x)e^{\psi(t,x)}\Big|<\infty,
\end{gathered}
\end{equation}
where the prime at the function $\psi$ means the derivative with respect to $x$ and the dot denotes the derivative with respect to $t$.

Suppose that the function $\psi(t,x)$ is a given slowly varying function of $t$ with the characteristic time of variation $\tau_a<T$ (the adiabaticity), i.e., $\psi(t,x)=q(t/\tau_a,x)$, where the characteristic scale of variation of the function $q(\tau,x)$ with respect to the dimensionless variable $\tau$ and the magnitude of the function $q(\tau,x)$ are of the order of unity. Such a representation of $\psi(t,x)$ allows one simply to deduce the quantity order with respect to $\tau_a$ provided this quantity is expressed in terms of $\psi(t,x)$. If one chooses
\begin{equation}\label{f_expl}
    f(t,x)=\int_{-T/2}^t d\tau \psi'(\tau,x)e^{\psi(\tau,x)-\psi(t,x)},
\end{equation}
then the metric $g_{\mu\nu}$ is obtained from the Minkowski metric by making the change of coordinates
\begin{equation}\label{ch_coord}
    t'=-T/2+\int_{-T/2}^td\tau e^{\psi(\tau,x)},\qquad x'=\int_0^x dye^{\psi(t,y)},
\end{equation}
up to the terms of the order $O(\tau_a^{-1})$ (see below). It is clear that $f(t,x)\rightarrow 0$ for $|x|\rightarrow\infty$. In view of \eqref{bound_cond_psi}, \eqref{f_expl}, the requirement $f(T/2,x)=0$ leads to the restriction
\begin{equation}
    \Big[\int_{-T/2}^{T/2}d\tau e^{\psi(\tau,x)}\Big]'=0\;\;\Rightarrow\;\;\int_{-T/2}^{T/2}d \tau e^{\psi(\tau,x)}=T.
\end{equation}
This restriction is satisfied by the class of smooth functions
\begin{equation}
\begin{gathered}\label{phi_class}
    e^{\psi(t,x)}=:1+\vf(t,x),\qquad \vf(t,x)=-\vf(-t,x),\\
    |\vf(t,x)|<1,\qquad \vf(t,x)\underset{|t|\geq T/2}{=}0,\qquad \vf(t,x)\underset{|x|\rightarrow\infty}{\rightarrow}0,\qquad \vf'(t,x)\underset{|x|\rightarrow\infty}{\rightarrow}0,\\
    \Big|\int_{-\infty}^\infty dx\dot{\vf}(t,x)\Big|<\infty,
\end{gathered}
\end{equation}
where the conditions \eqref{bound_cond_psi} have been taken into account. In general, one can replace the condition of oddness of $\vf(t,x)$ with respect to $t$ by $\vf(t,x)=\dot{\Omega}(t,x)$, where $\Omega(t,x)$ is a smooth function vanishing for $|t|>T/2$ and such that the other conditions in \eqref{phi_class} are satisfied, but the class \eqref{phi_class} will be sufficient for our purpose. The metric components $g_{\mu\nu}$ are written in terms of the function $\vf(t,x)$ as
\begin{equation}\label{metr_comps}
    g_{00}(t,x)=[1+\vf(t,x)]^2,\qquad f(t,x)=\int_{-T/2}^t\frac{d\tau\vf'(\tau,x)}{1+\vf(t,x)}.
\end{equation}
Therefore, the characteristic time of variation of the eigenvalues $\omega_\al(t)$ of the quantum-field Hamiltonian \eqref{Hamilt_reg} is $\tau_a$. In a general position, the function $f(t,x)$ determined by the expression \eqref{metr_comps} is large and proportional to $\tau_a$ since $\vf'(t,x)$ is slowly varying on the time scales less than $\tau_a$.

The complete expression of the metric resulting from the change of coordinates \eqref{ch_coord} has the form
\begin{equation}\label{metrics_exct}
    ds^2=g_{\mu\nu}(t,x)dx^\mu dx^\nu-g_{00}(t,x)\big[h^2(t,x)dt^2+2h(t,x)dtdx \big],\qquad h(t,x):=\int_{0}^x\frac{dy\dot{\vf}(t,y)}{1+\vf(t,x)},
\end{equation}
where the components of the metric $g_{\mu\nu}$ are given in formulas \eqref{metrics}, \eqref{metr_comps}. If $b^{-1}$ is a characteristic scale of variation of $\vf(t,x)$ with respect to $x$ then
\begin{equation}\label{h_estim}
    h(t,x)=O\big((b\tau_a)^{-1}\big),
\end{equation}
uniformly with respect to $t$ and $x$. The metric \eqref{metrics_exct} goes to the Minkowski metric at $|x|\rightarrow\infty$ only if
\begin{equation}\label{phi_class_extr}
    \int_0^{+\infty}dx\vf(t,x)=\int_{-\infty}^0dx\vf(t,x)=0.
\end{equation}
However, we shall not impose the additional condition \eqref{phi_class_extr} on $\vf(t,x)$ since, in this case, it is rather hard to find a class exactly solvable potentials of the Schr\"{o}dinger equation \eqref{KG_stat_2} that satisfy \eqref{phi_class}, \eqref{phi_class_extr} even approximately. The difference of the metric $g_{\mu\nu}$ from the Minkowski metric at the spatial infinity is of the order of $(b\tau_a)^{-1}$ and it can be made arbitrarily small at $\tau_a\rightarrow+\infty$. Notice that if $\vf(t,x)$ tends exponentially to zero at $|x|\rightarrow\infty$ then $h'(t,x)$ and all the higher derivatives of $h(t,x)$ with respect to $x$ also tend exponentially to zero at $|x|\rightarrow\infty$.

It is clear that the metric always enters the effective action and scalars constructed with the aid of this metric in the combination \eqref{metrics_exct}. Therefore, we can neglect the second contribution to the metric \eqref{metrics_exct}, which is proportional to $h$, in comparison with the first one for $b\tau_a\gg1$. The change of the eigenfunctions and the eigenvalues of the equation \eqref{KG_stat} due to the contributions proportional to $h$ is also negligibly small at large $\tau_a$. Henceforth, we disregard these terms in the metric \eqref{metrics_exct} and consider only the metric $g_{\mu\nu}$.

Let us now estimate the contribution of the first term to the one-loop effective action \eqref{inout_eff_act}. We give here the two arguments demonstrating that this contribution can be neglected at sufficiently large $\tau_a$. The first argument appeals to the general theory of quantum transitions under the influence of adiabatic perturbations, which is given in \cite{LandLifshQM.11}, Sect. 53. According to this theory, the probability of transition from the state $1$ to the state $2$ is determined, in the leading order, by the nearest to the real axis singular point of the function $\omega_{21}(t)=\omega_2(t)-\omega_1(t)$ as a function of the complex variable $t$. In our case, $\omega_1(t)$ is the energy of the instantaneous vacuum state (the last term in \eqref{Hamilt_reg}) and $\omega_2(t)$ is the energy of the first excited state (the minimal energy of the one-particle state). The functions $\omega_{1,2}(t)$ possess a singularity of the form $\sqrt{t-t_0}$ at the point $t_0\in \mathbb{C}$ where the terms cross: $\omega_{1}(t_0)=\omega_2(t_0)$ (for details see, e.g., \cite{JoyeThes, HagJoye}). Since the functions $\omega_{1,2}(t)$ do not have singularities on the real axis and possess a characteristic scale of variation of the order $\tau_a$ then $\im t_0\gtrsim\tau_a$. Consequently, we have the following rough estimate for the transition probability from the instantaneous vacuum state to the first excited state over the period of time $\tau_a$:
\begin{equation}\label{Landau-Zener_est}
    w_{21}\sim\exp\Big(-2m\int_0^{t_0}dt\sqrt{1-t/t_0}\Big)\sim\exp(-4m\tau_a/3).
\end{equation}
The quantum transitions from the vacuum to the states with higher energies are even more suppressed. Therefore, taking sufficiently large $\tau_a$, we can make this probability to be arbitrarily small. The transitions from the instantaneous vacuum state are the result of the action of the two last terms in \eqref{Smatr_energy_1}, i.e., they originate from the action of the operators constructed using the quadratic form $A_{\al\be}$ in the notation of \eqref{Hamil_gener}. Thus, at sufficiently large $\tau_a$, we can neglect the contributions of these operators to the quantum-field Hamiltonian. Then the contribution of the first term to the one-loop effective action \eqref{inout_eff_act} vanishes (see the remark after \eqref{vac_vac_expl}). As a rule, the estimations of the form \eqref{Landau-Zener_est} are valid in the adiabatic limit for the Hamiltonians analytic in $t$ near the real axis. The sufficient conditions that, in the adiabatic limit, formula \eqref{Landau-Zener_est} holds can be found in \cite{JoyeThes, JoyeGener, Nenciu, ElgHag}.

The second argument is based on the direct evaluation of the matrix elements of the operator $\bar{A}_{\al\be}$. Indeed, the skew-scalar product defining $\bar{A}_{\al\be}$ is written as
\begin{equation}
    \{\dot{\ups}_\al(t),\ups_\be(t)\}=\int d x
    \left[
      \begin{array}{cc}
        \dot{u}_\al & \dot{w}_\al \\
      \end{array}
    \right]
    \left[
      \begin{array}{cc}
        0 & -1 \\
        1 & 0 \\
      \end{array}
    \right]
    \left[
      \begin{array}{c}
        u_\be \\
        w_\be \\
      \end{array}
    \right]=
    \int dx(\dot{w}_\al u_\be-\dot{u}_\al w_\be).
\end{equation}
This expression should be also symmetrized over $\al$, $\be$ (see \eqref{Smatr_energy_1}). It follows from \eqref{KG_stat_2} that
\begin{equation}\label{u_w_approx}
\begin{split}
    \dot{u}_\al&=e^{iS_\al}\Bigl[\dot{\tilde{u}}_\al-i\Bigl(\dot{\omega}_\al\int_0^xdyf+\omega_\al\int_0^xdy\dot{f}\Bigr)\tilde{u}_\al\Bigr],\\
    \dot{w}_\al&= e^{iS_\al}\Big[\dot{f}\partial_x\tilde{u}_\al+f\partial_x\dot{\tilde{u}}_\al-i\dot{\omega}_\al\tilde{u}_\al-i\omega_\al\dot{\tilde{u}}_\al -i\Bigl(\dot{\omega}_\al\int_0^xdyf+\omega_\al\int_0^xdy\dot{f}\Bigr)(f\partial_x\tilde{u}_\al-i\omega_\al\tilde{u}_\al)\Big].
\end{split}
\end{equation}
Then
\begin{equation}\label{A_saddle_point}
    \{\dot{\ups}_{\al}(t),\ups_{\be}(t)\}=\int dx e^{i(S_\al+S_\be)}[\cdots ].
\end{equation}
The expression standing in the exponent in the integrand is large since it is proportional to $m\tau_a$. Therefore, the integral over $x$ can be estimated with the help of the stationary phase method. The stationary points are located at zeros of the function $f(t,x)$ as a function of the variable $x$. We suppose that $f(t,x)$ possesses a single simple root at the point  $x=0$ for any $t$, viz.,
\begin{equation}\label{extrem_restr}
    \vf'(t,0)=0,\qquad\vf''(t,0)\neq0,\quad\forall t\in \mathbb{R}.
\end{equation}
It is clear that
\begin{equation}
    \tilde{u}_\al=O(1),\qquad \partial_x\tilde{u}_\al=O(1),\qquad \omega_\al=O(1),
\end{equation}
for $\tau_a\rightarrow+\infty$. The expression in the square brackets in \eqref{A_saddle_point} contains either the terms $O(\tau_a^{-1})$ or the terms $O(1)$ vanishing at $x=0$. Applying to \eqref{A_saddle_point} the standard formula for the contribution of the saddle point, we arrive at
\begin{equation}\label{skew_approx}
    \{\dot{\ups}_{\al}(t),\ups_{\be}(t)\}=O(\tau_a^{-3/2}).
\end{equation}
Loosely speaking, the integral \eqref{A_saddle_point} is small at large $\tau_a$ since, in this case, the integrand is a rapidly oscillating function of $x$. Notice that if the expression standing in the exponent in \eqref{A_saddle_point} had not have the stationary point then, according to the Riemann-Lebesgue lemma, the integral \eqref{A_saddle_point} would be decreasing faster than any power of $\tau_a$ at $\tau_a\rightarrow+\infty$.

Now we see that even if we take into account that the one-loop correction contains the operator $A_{\al\be}$ integrated over time (see the first correction \eqref{1_corr_nonstat}) that may result in appearance of the additional factor $\tau_a$, we can neglect the contributions of the operator $A_{\al\be}$ to the one-loop effective action in the limit $\tau_a\rightarrow+\infty$. In other words, the first term in the one-loop effective action \eqref{inout_eff_act} vanishes in this limit, and we have to investigate only the contribution of the second term in \eqref{inout_eff_act}.

\section{Schr\"{o}dinger equation}\label{Schr_Equation}

The equation for the mode functions \eqref{KG_stat_2} is useful to rewrite as the one-dimensional Schr\"{o}dinger equation,
\begin{equation}\label{Schrod_eqn}
    \big[-\partial_x^2+m^2((1+\vf)^2-1)\big]\tilde{u}_\al=(\omega^2_\al-m^2)\tilde{u}_\al=:\lambda\tilde{u}_\al,
\end{equation}
with the potential vanishing at the spatial infinity. We assume that the function $\vf(t,x)$ is infinitely differentiable and rapidly tends to zero at the spatial infinity with all its derivatives. In particular, the condition \eqref{potential_cond} is fulfilled. Of course, the requirements \eqref{phi_class}, \eqref{extrem_restr} are also supposed to be satisfied.

The contribution of the second term to \eqref{inout_eff_act} can be cast into the form (in many respects, we follow here the procedure presented in \cite{Bordag_scat,Khusnutd})
\begin{equation}\label{one_loop_st}
    \Ga^{(1)}_{st}\equiv-\int_{-T/2}^{T/2} dt E_0(t)=-\frac12\int_{-T/2}^{T/2} dt\int_{-m^2}^{\La^2-m^2} d\la \rho(\la)\sqrt{\la+m^2},
\end{equation}
where $\rho(\la)$ is the spectral density \eqref{spectr_denst} of the Schr\"{o}dinger equation operator \eqref{Schrod_eqn}. It is assumed that for $\la>0$ the mode functions are subject to the zero boundary conditions at $|x|=L/2$, where $L$ is much bigger than the characteristic scale of variation of $\vf(t,x)$ as a function of $x$. Let us split the spectral density onto two parts
\begin{equation}
    \rho(\la)=\rho_0(\la)+\De\rho(\la),
\end{equation}
where $\rho_0(\la)$ is the spectral density corresponding to $\vf=0$ and $\De\rho(\la)$ is given by formula \eqref{spectr_denst_dif}. The contribution of $\rho_0(\la)$ to $E_0(t)$ takes the form of the non-renormalized Coleman-Weinberg potential \cite{ColWein.9},
\begin{multline}
    \frac12\int_{-m^2}^{\La^2-m^2} d\la \rho_0(\la)\sqrt{\la+m^2}=\\
    =\frac{1}{2}\int_{-L/2}^{L/2}dx\int_{-\infty}^{\infty}\frac{dp}{2\pi}\sqrt{m^2+p^2}\theta(\La-\sqrt{m^2+p^2})
    \underset{\La\rightarrow+\infty}{\rightarrow} \int_{-L/2}^{L/2} dx\Big[\frac{\La^2}{4\pi}+\frac{m^2}{8\pi}\ln\frac{4\La^2}{m^2e}\Big],
\end{multline}
for the spacetime with the metric $\eta_{\mu\nu}$.

In virtue of the restrictions \eqref{phi_class} on the function $\vf(t,x)$, there is such $\la_0>-m^2$ that all the eigenvalues $\la>\la_0$. Then the contribution of $\De\rho(\la)$ to the instantaneous energy of zero-point fluctuations can be written as
\begin{equation}\label{energy_zero_p}
    \De E_0(t)=\frac12\int_{-m^2}^{\La^2-m^2} d\la \De\rho(\la)\sqrt{\la+m^2}=\int_{\la_0}^{\La^2-m^2} \frac{d\la}{4\pi i}\partial_\la\ln\frac{t(\sqrt{\la_+})}{t(\sqrt{\la_-})} \sqrt{\la+m^2},
\end{equation}
where the branch of the square root standing in the transition amplitude is taken with the cut along the real positive semiaxis, while for $\sqrt{\la+m^2}$ the principal branch is chosen. The integral in \eqref{energy_zero_p} is useful to represent as the contour integral
\begin{equation}
    \De E_0(t)=\frac{i}{4\pi}\int_{C_\La}d\la\partial_\la\ln t(\sqrt{\la})\sqrt{\la+m^2},
\end{equation}
where $C_\La$ starts at the point $(\La^2-m^2+i0)$, encircles counter-clockwise the point $\la_0$, and ends at the point $(\La^2-m^2-i0)$. For large $\la$, the transition amplitude has the asymptotics (see [(25), \cite{FaddZakh}])
\begin{equation}
    \ln t(\sqrt{\la})\underset{\la\rightarrow\infty}{\rightarrow}-\frac{c_1}{\sqrt{\la}}+O(\la^{-3/2}),\qquad c_1=\frac{im^2}{2}\int_{-\infty}^\infty dx(g_{00}-1)=\frac{im^2}{2}\int_{-\infty}^\infty dx(\sqrt{|g|}-1).
\end{equation}
Then
\begin{equation}
    \De E_0(t)=\frac{i}{4\pi}\int_{C_\La}d\la\Big(\partial_\la\ln t(\sqrt{\la})-\frac{c_1}{2\la^{3/2}}\Big)\sqrt{\la+m^2}+\frac{ic_1}{8\pi}\int_{C_\La}d\la\la^{-3/2}\sqrt{\la+m^2}.
\end{equation}
The first integral allows for the regularization removal limit $\La\rightarrow+\infty$. The second integral is reduced to
\begin{multline}
    \int_{C_\La}d\la\la^{-3/2}\sqrt{\la+m^2}=\int_{C_\La}d\la\la^{-3/2}(\sqrt{\la+m^2}-m)+m\int_{C_\La}d\la\la^{-3/2}\approx\\
    \approx-2\int_0^{\La^2-m^2}d\la\la^{-3/2}(\sqrt{\la+m^2}-m)\approx -2\ln\frac{4\La^2}{m^2e^2},
\end{multline}
where in the approximate equalities, the terms vanishing in the regularization removal limit are discarded. As a result, neglecting the terms vanishing at $L\rightarrow+\infty$, we have
\begin{equation}\label{energy_zero1}
    E_0(t)=\int_{-L/2}^{L/2}dx\Big\{\sqrt{|g|}\Big[\frac{\La^2}{8\pi g_{00}}+\frac{m^2}{8\pi}\ln\frac{4\La^2}{m^2e^2}\Big]+\frac{m^2}{8\pi}\Big\} +\frac{i}{4\pi}\int_{C_\infty}d\la\Big(\partial_\la\ln t(\sqrt{\la})-\frac{c_1}{2\la^{3/2}}\Big)\sqrt{\la+m^2}.
\end{equation}
The structure of divergencies coincides with the general expression found in \cite{KalKaz.11,gmse.11}. The conformal anomaly (the coefficient at $\ln m^2$) also coincides with the well-known answer for the flat spacetime. The last term in the curly brackets is not written in the explicitly covariant form in terms of $g_{\mu\nu}$ and the vector field $\xi^\mu=(1,0)$. The contributions of such a type standing at $m^2$ in the large mass expansion are also contained in the last term in \eqref{energy_zero1}. Only in sum with these contributions does the coefficient at $m^2$ take the covariant form (for more details, see \cite{KalKaz.11} for the case of four-dimensional spacetime).

It is useful to rewrite the last integral in \eqref{energy_zero1} as the integral over the cut of the function $\sqrt{\la+m^2}$. The expression in the round brackets in \eqref{energy_zero1} is holomorphic in the complex $\la$ plane with the cut along $[\la_1,+\infty)$, where $\la_1$ is the lowest eigenvalue of \eqref{Schrod_eqn}, and is real at $\la\in(-\infty,\la_1)$ \cite{Faddeev} (see also \cite{ZMNP}, Sect. 1.1). At large $\la$, it decreases as $\la^{-5/2}$. Therefore, we can deform the contour $C_\infty$ such that it encompasses the principal branch cut of $\sqrt{\la+m^2}$ and reduce the integral to the integral over the cut. Then we obtain (cf. \cite{Bordag_scat,Khusnutd})
\begin{equation}\label{energy_zero_fin}
    \frac{i}{4\pi}\int_{C_\infty}d\la\Big(\partial_\la\ln t(\sqrt{\la})-\frac{c_1}{2\la^{3/2}}\Big)\sqrt{\la+m^2}=\int_m^\infty\frac{dk}{2\pi}\Big(\partial_k\ln t(ik) +\frac{ic_1}{k^2}\Big)\sqrt{k^2-m^2}.
\end{equation}
The expression in the round brackets is real on the integration interval. It is the contribution \eqref{energy_zero_fin} to the one-loop effective action that we are interested in in what follows. In fact, we need to show that there exists the potential in the class of the potentials considered such that this contribution is not expandable in the convergent Laurent series in $m^{-2}$ after the integration over $t$. Otherwise, in view of the non-renormalizability of quantum gravity, the contribution \eqref{energy_zero_fin} can be always canceled out by the local counterterms.

Let us show that the term in the square brackets in \eqref{energy_zero1} contains all the terms at $m^2\ln m^2$ and also find the contribution to the one-loop effective action standing at $m^{-2}$ in the large mass expansion. To this end, we employ the general formula [(34), \cite{KalKaz.11}] for the high-temperature expansion of the free energy of fermions. In accordance with [(9), \cite{KalKaz.11}] (see also \cite{olopqfeces.9}) and [(34), \cite{KalKaz.11}], the instantaneous energy of zero-point fluctuations takes the form
\begin{equation}\label{energy_zero_asym}
    E_0(t)=\partial_\be(\be F_f)=\zeta(0)\s^0_\epsilon+\cdots=-\frac12\s^0_\epsilon+\cdots,
\end{equation}
where $\be\rightarrow0$ is the regularization parameter, the dots denote the terms standing at the non-negative powers of $m^2$, the complex parameter $\epsilon$ must be set to zero in the final answer, and
\begin{equation}\label{sigma_nu}
    \s_\nu^0=\int_0^\infty d\omega\int_C\frac{d\tau\tau^{\nu-1}}{2\pi i}\Sp e^{\tau(-\partial_x^2+m^2g_{00}-\omega^2)}.
\end{equation}
The integration contour $C$ goes top down parallel to the imaginary axis and slightly to the left of it.

In order to obtain the asymptotic expansion in $m^{-2}$, we employ the asymptotic expansion of the trace of the heat kernel entering into \eqref{sigma_nu}:
\begin{equation}\label{HK_expan}
    \Sp e^{\tau(-\partial_x^2+m^2g_{00}-\omega^2)}=\int_{-L/2}^{L/2} dx e^{i\pi d/2}\frac{a_k(t,x)}{(4\pi)^{d/2}}\tau^{k-d/2}e^{-\tau(\omega^2-\tilde{m}^2)},
\end{equation}
where, for brevity, we have introduced the notation $\tilde{m}^2:=m^2g_{00}$ and $d=1$. The cut of $\tau^{k-d/2}$ is chosen to be along the real positive semiaxis. The resummation analogous to [(18), \cite{KalKaz.11}] has been made in \eqref{HK_expan}, so $\tilde{m}^2$ without derivatives does not appear in the expansion coefficients $a_k$. Also the terms vanishing in the limit $L\rightarrow+\infty$ have been thrown away from \eqref{HK_expan}. The coefficients $a_k$ are representable in the form
\begin{equation}\label{akj}
    a_k=\sum_{j=0}^{[4k/3]}a_k^{(j)}m^j,
\end{equation}
where the sum is carried over even $j$ only. The upper summation limit is determined in the same way as in formula [(38), \cite{KalKaz.11}] in counting the powers of $\omega$. Substituting the expansion \eqref{HK_expan} into \eqref{sigma_nu} and integrating, at first, over $\tau$ and then over $\omega$, we obtain the asymptotic expansion
\begin{equation}
    \s^0_\nu=e^{i\pi\nu}\int_{-L/2}^{L/2} dx\sum_{k=0}^\infty(-1)^k\frac{a_k}{(4\pi)^{D/2}}\Ga(k+\nu-D/2)(\tilde{m}^2)^{D/2-\nu-k},
\end{equation}
where $D:=d+1$. Writing $a_k$ in the form \eqref{akj} and collecting the terms at the same power of $m^2$, we come to
\begin{equation}
    \s^0_\nu=e^{i\pi\nu}\int_{-L/2}^{L/2} dx\sum_{s=0}^\infty\sum_{n=0}^{2s}\frac{(m^2)^{D/2-s-\nu}}{(4\pi)^{D/2}}(-1)^{s+n}\Ga(s+n+\nu-D/2)g_{00}^{D/2-\nu-n-s}a_{s+n}^{(2n)}.
\end{equation}
The explicit expressions for $a_k^{(j)}$ can be readily found from the general formulas given in Appendix B of \cite{KalKaz.11}. At that, it should be taken into account that in our case
\begin{equation}
    E=\omega^2-m^2g_{00},\qquad V=0,\qquad \Omega_{\mu\nu}=0,\qquad \tilde{R}^\mu_{\ \nu\rho\s}=0.
\end{equation}

It follows from formula [(34), \cite{KalKaz.11}] that the terms containing $\ln m^2$ stem only from the expansion of $\s_\nu^0$ in $\nu$ in the neighbourhood of $\nu=0$. Recall that, in \cite{KalKaz.11}, the parameter $\nu$ enters into the definition of $d$ and the common factor $e^{i\pi\nu}$ for the both terms in [(34), \cite{KalKaz.11}] has been already put to unity. In our case, the terms containing $\ln m^2$ arise if
\begin{equation}
    D/2-s-n=0,1,2,\ldots,
\end{equation}
i.e., for $s=0$, $n=0$ and $s=1$, $n=0$. Taking into account that $a_0=1$ and $a_1=0$, and discarding the factor $e^{i\pi\nu}$, we have
\begin{equation}
    \s^0_\nu=\int_{-L/2}^{L/2} dx\Big[-\frac{\tilde{m}^2}{4\pi\nu}+\frac{\tilde{m}^2}{4\pi}\ln(\tilde{m}^2e^{\ga-1})+O(\nu)\Big],
\end{equation}
where $\ga$ is the Euler constant. Substituting this expansion into \eqref{energy_zero_asym}, we see that the coefficient at $m^2\ln m^2$ coincides exactly with that standing in formula \eqref{energy_zero1}. Other contributions containing $\ln m^2$ are absent in the asymptotic expansion in $m^{-2}$.

As for the term at $m^{-2}$ in $\s^0_0$, we obtain
\begin{multline}
    \int_{-L/2}^{L/2} \frac{dx}{4\pi}\Big[\frac{a_2^{(0)}}{g_{00}}-\frac{a_3^{(2)}}{g_{00}^2}+2\frac{a_4^{(4)}}{g_{00}^3}-3!\frac{a_5^{(6)}}{g_{00}^4}+4!\frac{a_6^{(8)}}{g_{00}^5}\Big]= \\ =\int_{-L/2}^{L/2} \frac{dx}{4\pi}\Big[\frac1{3!}\frac{E^{(4)}}{10m^2g_{00}^2}+\frac2{4!}\frac{\frac45E'E''+\frac13(E'')^2+\frac4{15}(E'')^2}{m^4 g_{00}^3}-\frac{3!}{5!}\frac{-2(E')^2E''-\frac53(E')^2E''}{m^6g_{00}^4}+\frac{4!}{6!}\frac{5(E')^4}{2m^8g_{00}^5} \Big].
\end{multline}
Collecting the terms, this expression can be written as
\begin{equation}
    \int_{-L/2}^{L/2} \frac{dx}{240\pi}\Big[-\frac{g_{00}^{(4)}}{g_{00}^2} +\frac{4g_{00}'g_{00}'''+3(g_{00}'')^2}{g_{00}^3} -11\frac{(g_{00}')^2g_{00}''}{g_{00}^4} +5\frac{(g_{00}')^4}{g_{00}^5}\Big]=\int_{-L/2}^{L/2} \frac{dx}{240\pi}\Big[\frac{(g_{00}'')^2}{g_{00}^3}-\frac53\frac{(g_{00}')^4}{g_{00}^5}\Big],
\end{equation}
where the terms vanishing at $L\rightarrow+\infty$ have been discarded in the last equality. If one substitutes $g_{00}(t,x)$ from \eqref{metr_comps} to the resulting expression and integrates over $t$ as in \eqref{one_loop_st} then, in a general position, this contribution does not vanish even if one takes into account the conditions \eqref{phi_class}. It may appear at first sight that this result contradicts the general statement proved in \cite{KalKaz2.12}, Sect. 3, that the coefficients at the negative powers of $m^2$ in the one-loop effective action on a stationary background with the standard definition of the Hamiltonian are expressed solely in terms of the spacetime metric, its curvature, and covariant derivatives of the curvature. In other words, one may expect that, in our case, the term in the one-loop effective action at $m^{-2}$ is zero. However, in the case at issue, the statement proved in \cite{KalKaz2.12} is inapplicable since the metric \eqref{metrics} depends on $t$ and is not stationary in the system of coordinates where the Hamiltonian is constructed.

Now we give a rough argument for that, in a general position, the contribution \eqref{energy_zero_fin} is not expandable in a convergent Laurent series in $m^{-2}$. This argument is the analog of the ``Dyson argument''. Indeed, as we have shown, the asymptotic expansion \eqref{energy_zero_fin} in a large mass does not contain the terms at $\ln m^2$ and includes only the powers of $m^{-2}$, the coefficients at which being real. If this series were convergent in some ring of the complex $m^2$ plane then the expression \eqref{energy_zero_fin} should remain real under the change $m^2\rightarrow-m^2$. Changing continuously $m^2$ from $m^2$ to $-m^2$, we encounter the situation (at $m^2=0$) when the contour $C_\infty$ in the $\la$ plane is pinched by the branch points of the function $\sqrt{\la+m^2}$ and the expression standing in the round brackets in \eqref{energy_zero_fin}. As a result, the contribution \eqref{energy_zero_fin} possesses a singular point at $m^2=0$ in the $m^2$ plane, and in passing through this point, \eqref{energy_zero_fin} acquires a complex additive. The bypass rule of the singular point is dictated by the prescription $m^2\rightarrow m^2-i0$, viz., the branch point of $\sqrt{\la+m^2}$ should lie a little bit higher than the real axis. When $m^2$ is set to $-m^2$, the branch point of $\sqrt{\la+m^2}$ is located slightly higher than the cut (corresponding to the continuous spectrum of the Schr\"{o}dinger equation) of the expression enclosed by the round brackets in \eqref{energy_zero_fin}. It is clear that, in a general position, \eqref{energy_zero_fin} has a non-zero imaginary part for such a configuration since all the terms in \eqref{energy_zero_fin} are complex and there is not any symmetry or property of $\ln t(\sqrt{\la})$ that prevents \eqref{energy_zero_fin} to be complex. From the physical point of view, the existence of the imaginary contributions to the effective action at $m^2<0$ is also obvious as this situation corresponds to the presence of a tachyon in the theory, i.e., the vacuum is unstable and $\im\Ga^{(1)}_{in-out}\neq0$.

Of course, this argument is not rigorous since, having integrated over $t$ in accordance with \eqref{one_loop_st}, the imaginary contribution to the effective action may disappear due to the properties of the potential \eqref{phi_class}. In the next section, we shall consider a concrete example and make sure that this cancelation does not happen.

\section{Example}\label{Example}

Unfortunately, it seems there is not the family of exactly solvable potentials of the form \eqref{Schrod_eqn} that satisfies all the properties \eqref{phi_class}, \eqref{extrem_restr} (see the exactly solvable potentials in \cite{BagGit1}). Therefore, we shall consider a certain family of potentials that are not exactly solvable, but satisfying \eqref{phi_class}, \eqref{extrem_restr} and approximated with a good accuracy by the exactly solvable potentials. The accuracy of this approximation will be controlled by the semiclassical methods.

Let
\begin{equation}\label{exampl_phi}
    \vf(t,x)=\frac{\bar{c}(t)}{\ch^2(b(t)x)},\qquad \bar{c}(t)=-\bar{c}(-t),\quad |\bar{c}(t)|<1,\qquad \infty>b(t)=b(-t)>0,
\end{equation}
and $\bar{c}(t)=0$ for $|t|\geq T/2$. The functions $b(t)$, $\bar{c}(t)$ are supposed to be smooth. It is clear that the functions \eqref{exampl_phi} satisfy the conditions \eqref{phi_class}, \eqref{extrem_restr}, and the estimate \eqref{h_estim}. It is not difficult to check that the estimate \eqref{skew_approx} is also fulfilled. The potential of the Schr\"{o}dinger equation \eqref{Schrod_eqn} corresponding to such functions $\vf(t,x)$,
\begin{equation}
    V(t,x)=m^2\Big[\frac{2\bar{c}(t)}{\ch^2(b(t)x)}+\frac{\bar{c}^2(t)}{\ch^4(b(t)x)}\Big],
\end{equation}
is not exactly solvable. As the approximation to this potential, we consider the exactly solvable potential
\begin{equation}\label{TP_pot}
    V_0(t,x)=m^2\frac{c(t)}{\ch^2(b(t)x)},
\end{equation}
taking $c(t)$ in such a way that $V$ and $V_0$ coincide in the extremum point
\begin{equation}\label{ccbar}
    c(t)=2\bar{c}(t)+\bar{c}^2(t).
\end{equation}
Then the relative error amounts to
\begin{equation}
    \Big|\frac{V-V_0}{V_0}\Big|=\Big|\frac{\bar{c}^2}{2\bar{c}+\bar{c}^2}\Big|\tah^2(bx)\lesssim\frac{|c|}{4},\quad x\in \mathbb{R},
\end{equation}
where in the last inequality, it is assumed that $c$ is small. The plots of the potentials and the relative error for different $c$ are given in Fig. \ref{pots_fig}. It is to be expected that for small $c$ the spectral density of the Hamiltonian with the potential $V$ coincides with the spectral density of the Hamiltonian with the potential $V_0$ in the leading nontrivial order in $c$, where $c$ and $\bar{c}$ are related according to \eqref{ccbar}.

\paragraph{Semiclassical spectral density.}

In order to give a quantitative estimate for the difference between the spectral density corresponding to the potential $V$ and the spectral density corresponding to the potential $V_0$, we can use the well-known semiclassical formulas (see, e.g., \cite{BalBlochSch}, Sect. 5). According to the Bohr-Sommerfeld quantization rule, the semiclassical spectral density is expressed through the adiabatic invariants
\begin{equation}
    I(\la)=2\int_{x_1}^{x_2}dx\sqrt{\la-V(x)},\qquad I^0(\la)=2\int_{\s_1}^{\s_2}dx\sqrt{\la-V_0(x)},
\end{equation}
where the possible points $x_1$, $x_2$, $\s_1$, and $\s_2$ are shown on Fig. \ref{pots_fig}. Let us redefine, for brevity, the variables $\la\rightarrow m^2\la$ and $x\rightarrow x/b$. Then $I$ and $I^0$ contain the common factor $\bar{m}:=m/b$. We eliminate this factor from the definitions of $I$ and $I_0$. Then, in the leading order in derivatives of the potential, the semiclassical spectral density is given by the sum of terms of the form (see [(5.7), \cite{BalBlochSch}])
\begin{equation}\label{quasi_spectr_dens}
    \frac{1}{2\pi mb}\im\big[\tau(\la_+)e^{ik\bar{m}I(\la_+)}\big],
\end{equation}
and, in fact, is a manifestation of the Bohr-Sommerfeld quantization rule. In \eqref{quasi_spectr_dens}, $\tau(\la):=I'(\la)$ denotes a period and $k$ is some non-negative integer number. The semiclassical approximation is applicable in this case since $\bar{m}\gg1$.

We can write approximately the adiabatic invariant as
\begin{equation}\label{adiab_inv_appr}
    I\approx I^0+\bar{c}^2\int_{\s_1}^{\s_2}\frac{dx\sh^2x}{\sqrt{\la-V_0}\ch^4x}.
\end{equation}
Both the first and second integrals on the right-hand side of \eqref{adiab_inv_appr} can be evaluated analytically. Denoting the second term in \eqref{adiab_inv_appr} as $\de I$, we have
\begin{equation}\label{adiab_inv_appr1}
    I^0_{23}=2\pi(\sqrt{-c}-\sqrt{-\la}),\qquad\de I_{23}=\frac{\pi\bar{c}^2}{2\sqrt{-c}}\Big(1-\frac{\la}{c}\Big),
\end{equation}
where the indices at $I$ and $\de I$ denote the numbers of the turning points. For the adiabatic invariants corresponding to the infinite motion, we obtain
\begin{equation}\label{adiab_inv_appr2}
\begin{split}
    I^0_{12}&=2\sqrt{\la}\bigg[\arch\frac{\sh\frac{\bar{L}}{2}}{\sqrt{\frac{c}{\la}-1}}-\sqrt{\frac{c}{\la}}\arth\frac{\sh\frac{\bar{L}}{2}}{\sqrt{\frac{\la}{c}\ch^2\frac{\bar{L}}{2}-1}} \bigg]\approx \sqrt{\la}\Big[\bar{L}-\ln\big(\frac{c}{\la}-1\big) -\sqrt{\frac{c}{\la}}\arth\sqrt{\frac{c}{\la}} \Big],\\
    \de I_{12}&=\frac{\bar{c}^2}{2\sqrt{c}}\bigg[\big(1-\frac{\la}{c}\big)\arth\frac{\sh\frac{\bar{L}}{2}}{\sqrt{\frac{\la}{c}\ch^2\frac{\bar{L}}{2}-1}} +\frac{\tah\frac{\bar{L}}{2}}{\ch\frac{\bar{L}}{2}}\sqrt{\frac{\la}{c}\ch^2\frac{\bar{L}}{2}-1} \bigg]\approx \frac{\bar{c}^2}{2\sqrt{c}}\Big[\big(1-\frac{\la}{c}\big)\arth\sqrt{\frac{c}{\la}}+\sqrt{\frac{\la}{c}} \Big],\\
    I^0_{14}&=4\sqrt{\la}\arsh\frac{\sh\frac{\bar{L}}{2}}{\sqrt{1-\frac{c}{\la}}}+4\sqrt{-c}\arctg\frac{\sh\frac{\bar{L}}{2}}{\sqrt{1-\frac{\la}{c}\ch^2\frac{\bar{L}}{2}}}\approx 2\sqrt{\la}\Big[\bar{L}-\ln\big(1-\frac{c}{\la}\big)\Big]+4\sqrt{-c}\arctg\sqrt{-\frac{c}{\la}},\\
    \de I_{14}&=\frac{\bar{c}^2}{\sqrt{-c}}\bigg[\big(1-\frac{\la}{c}\big)\arctg\frac{\sh\frac{\bar{L}}{2}}{\sqrt{1-\frac{\la}{c}\ch^2\frac{\bar{L}}{2}}} -\frac{\tah\frac{\bar{L}}{2}}{\ch\frac{\bar{L}}{2}}\sqrt{1-\frac{\la}{c}\ch^2\frac{\bar{L}}{2}} \bigg]\approx \frac{\bar{c}^2}{\sqrt{-c}} \Big[\big(1-\frac{\la}{c}\big)\arctg\sqrt{-\frac{c}{\la}}-\sqrt{-\frac{\la}{c}} \Big],
\end{split}
\end{equation}
where $\bar{L}:=bL$, and it is assumed $\la/c\ch^2(\bar{L}/2)\gg1$ in the approximate equalities. In the expressions \eqref{adiab_inv_appr1}, \eqref{adiab_inv_appr2}, the principal branches of the multivalued functions are chosen. The adiabatic invariants corresponding to the infinite motion are responsible for the oscillations of the spectral density \eqref{quasi_spectr_dens} with a variation of $L$. The contribution of such terms to the effective action vanishes at $L\rightarrow+\infty$. As for $I_{23}$, we deduce
\begin{equation}
    \frac{\de I_{23}}{I^0_{23}}=-\frac{\bar{c}^2}{4c}\Big(1+\sqrt{\frac{\la}{c}}\Big),\qquad \Big|\frac{\de I_{23}}{I^0_{23}}\Big|\leq\Big|\frac{\bar{c}^2}{2c}\Big|\approx\frac{|c|}{8},
\end{equation}
where $c<0$, $\la\in(c,0)$, and it is supposed that $c$ is small. The same estimate holds also for $c>0$, $\la\in(0,c)$. At large $\la$, the approximation \eqref{adiab_inv_appr} becomes wrong since $x_{2,3}(\la)$ tend to the singular points of $\ch^{-2}x$. However, it follows from formula \eqref{adiab_inv_asymp} that
\begin{equation}
    \frac{I_{23}-I_{23}^0}{I^0_{23}}=o(1),
\end{equation}
for $\la\rightarrow\infty$ because the points $y_{2,3}$ coincide for the potentials $V(x)$ and $V_0(x)$. The fact that the turning points $x_{2,3}$ and $x_{2,3}^0$ of the adiabatic invariants $I_{23}$ and $I_{23}^0$ come to the singular points $y_{2,3}$ at $|\la|\rightarrow\infty$ is easy to see from the plots of $\im V$ and $\im V_0$ presented in Fig. \ref{pots_fig}. For real $\la$, the turning points move along the level lines of $\im V=0$ and $\im V_0=0$.

Thus we can state that the spectral density of the Hamiltonian with the potential $V$, its smooth part (the Thomas-Fermi contribution), and the oscillating part with the leading exponentially suppressed corrections are well approximated in the leading non-trivial order in the small parameter $c$ by the spectral density of the Hamiltonian with the potential $V_0$. Consequently, in evaluating the integral over $\la$ in \eqref{one_loop_st}, we can use the spectral density corresponding to the Hamiltonian with the potential $V_0$.

\begin{figure}[t]
\newcounter{lwdth}\setcounter{lwdth}{33564238}
\newlength{\lwdtl}\setlength{\lwdtl}{33564238 sp}
\newcounter{wdth}\setcounter{wdth}{\value{lwdth}}\setcounter{wdth}{\value{wdth}/65536}
\newcounter{hght}\setcounter{hght}{\value{lwdth}}\setcounter{hght}{\value{wdth}/65536}\setcounter{hght}{\value{wdth}*\real{0.5}}
\newcounter{wdth1}\setcounter{wdth1}{\value{wdth}*\real{0.5}}\setcounter{wdth1}{\value{wdth1}+17}
\newcounter{hght1}\setcounter{hght1}{\value{hght}*\real{0.5}}
\newcounter{hght2}\setcounter{hght2}{\value{hght}-15}
\centering
\begin{picture}(\value{wdth},\value{hght2})(-34,-10)
\put(0,0){\mbox{\includegraphics*[width=0.45\lwdtl ]{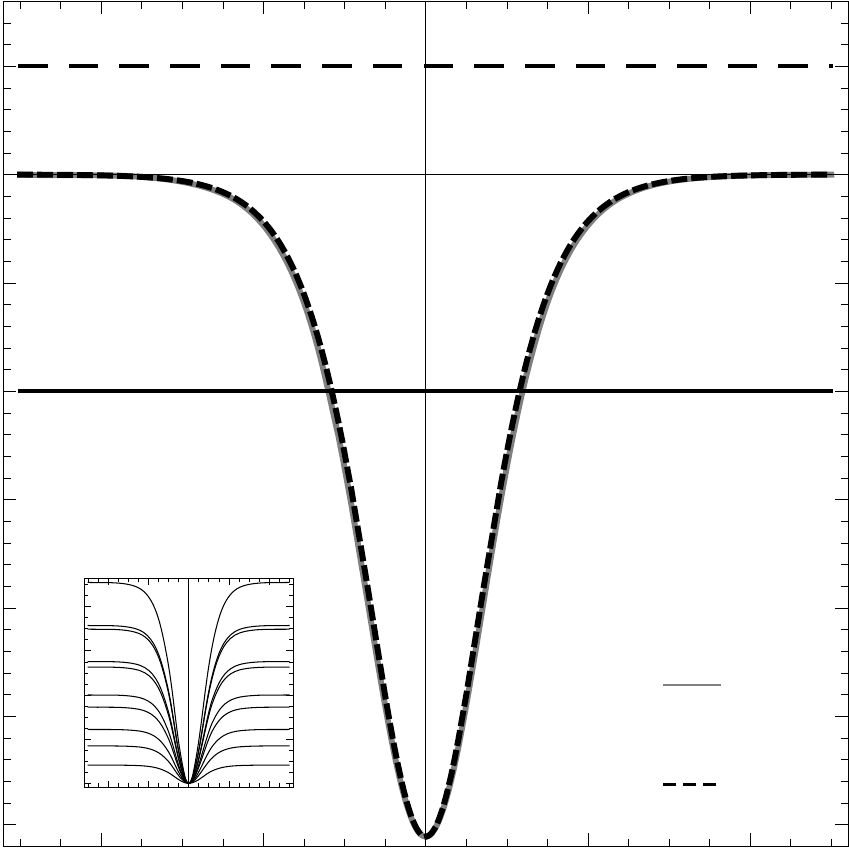}}}
\put(\value{wdth1},\value{hght1}){\mbox{\includegraphics*[width=0.4\lwdtl]{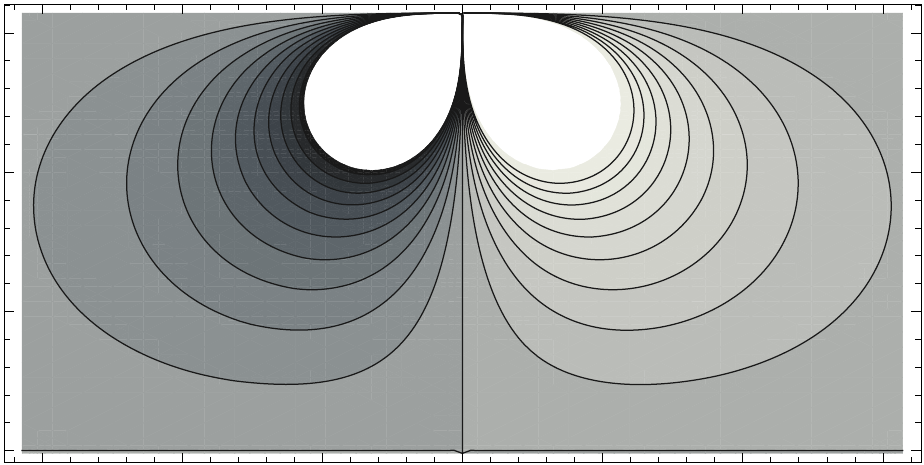}}}
\put(\value{wdth1},0){\mbox{\includegraphics*[width=0.4\lwdtl]{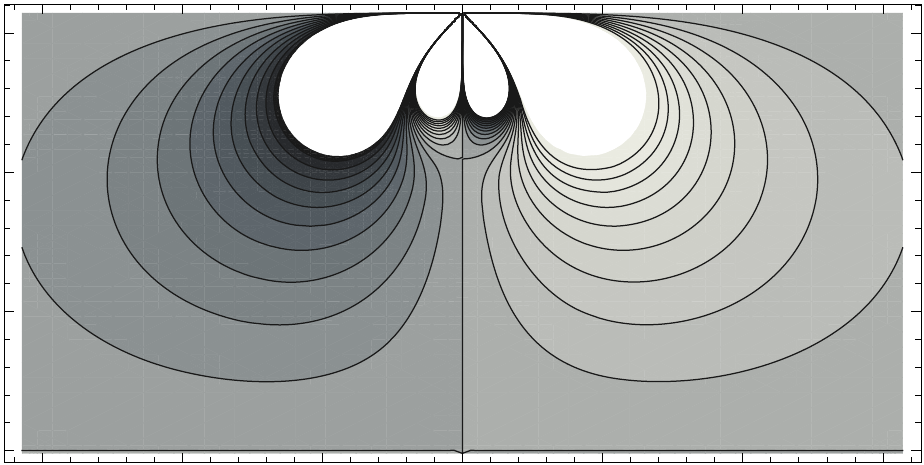}}}
\put(200,42){\mbox{$V(x)$}}
\put(200,15){\mbox{$V_0(x)$}}
\put(80,42){\mbox{$\frac{\delta V}{V}$}}
\put(76,112){\mbox{$x_2$}}
\put(145,112){\mbox{$x_3$}}
\put(126,131){\mbox{$\sigma_3$}}
\put(95,131){\mbox{$\sigma_2$}}
\put(6,203){\mbox{$\sigma_1$}}
\put(6,218){\mbox{$x_1$}}
\put(214,218){\mbox{$x_4$}}
\put(214,203){\mbox{$\sigma_4$}}
\put(113,-10){\mbox{$0$}}
\put(156.6,-10){\mbox{$2$}}
\put(200.2,-10){\mbox{$4$}}
\put(60,-10){\mbox{$-2$}}
\put(15.5,-10){\mbox{$-4$}}
\put(48.5,10){\mbox{\scriptsize$0$}}
\put(70.5,10){\mbox{\scriptsize$4$}}
\put(20.5,10){\mbox{\scriptsize$-4$}}
\put(6.5,15){\mbox{\scriptsize$0.00$}}
\put(6.5,39){\mbox{\scriptsize$0.04$}}
\put(6.5,63.5){\mbox{\scriptsize$0.08$}}
\put(-30,3.2){\mbox{$-0.30$}}
\put(-30,62){\mbox{$-0.20$}}
\put(-30,120.8){\mbox{$-0.10$}}
\put(-21.5,179.6){\mbox{$0.00$}}
\put(-21.5,209){\mbox{$0.05$}}
\put(-40,112){\mbox{a)}}
\put(258,-0.4){\mbox{$0.0$}}
\put(258,30.7){\mbox{$0.5$}}
\put(258,61){\mbox{$1.0$}}
\put(258,91.7){\mbox{$1.5$}}
\put(247,45.7){\mbox{c)}}
\put(368.85,-10){\mbox{$0.0$}}
\put(430.5,-10){\mbox{$1.0$}}
\put(298.8,-10){\mbox{$-1.0$}}
\put(464,91){\circle{13}}
\put(461,89){\mbox{$x$}}
\put(258,128){\mbox{$0.0$}}
\put(258,158.7){\mbox{$0.5$}}
\put(258,189){\mbox{$1.0$}}
\put(258,220.1){\mbox{$1.5$}}
\put(247,174.1){\mbox{b)}}
\put(368.85,118.4){\mbox{$0.0$}}
\put(430.5,118.4){\mbox{$1.0$}}
\put(298.8,118.4){\mbox{$-1.0$}}
\put(464,219.4){\circle{13}}
\put(461,217.4){\mbox{$x$}}
\end{picture}
\caption{{\footnotesize Panel a): The potentials $V(x)$ and $V_0(x)$ at $L=10$, $m=1$, $b=1$, $\bar{c}=-1/6$, and $c$ given in \eqref{ccbar}. The turning points of $V_0(x)$ are denoted by $\s_k$, while the turning points of the potential $V(x)$ are $x_k$. Inset: The relative error $\de V/V$ for $\bar{c}$ changing in the interval from $-1/6$ to $1/6$ with the step $1/10$, and $c$ given in \eqref{ccbar}. Panels b) and c): The level lines of $\im V_0(x)$ and $\im V(x)$, respectively, at $m=1$, $b=1$, $\bar{c}=-1/4$, and $c$ given in \eqref{ccbar}.}}
\label{pots_fig}
\end{figure}
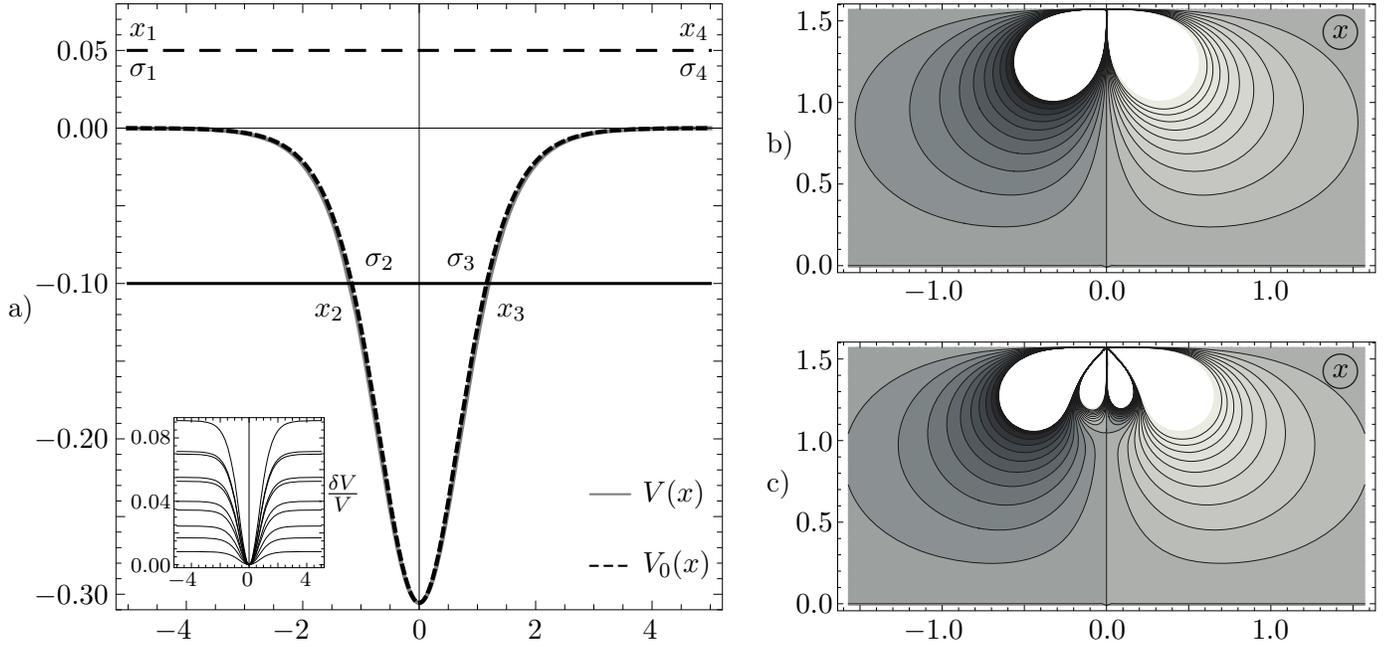

\paragraph{Exact solution.}

Let us consider now the Schr\"{o}dinger equation \eqref{Schrod_eqn} with the potential \eqref{TP_pot}. The exact solution of this equation needed for us is given in the problem 4, \cite{LandLifshQM.11}, Sect. 25 (see also \cite{BagGit1}):
\begin{equation}\label{exact_sol}
\begin{split}
    \psi&=\Big(\frac{1-\xi^2}{4}\Big)^{-\frac{ik}{2b}}F\Big(-\frac{ik}{b}-s,1-\frac{ik}{b}+s;1-\frac{ik}{b};\frac{1-\xi}{2}\Big)=\\
    &=\Big(\frac{1-\xi^2}{4}\Big)^{-\frac{ik}{2b}}\frac{\Ga\big(1-\frac{ik}{b}\big)\Ga\big(\frac{ik}{b}\big)}{\Ga(-s)\Ga(s+1)}F\Big(-\frac{ik}{b}-s,1-\frac{ik}{b}+s;1-\frac{ik}{b};\frac{1+\xi}{2}\Big)\\ &+\Big(\frac{1+\xi}{1-\xi}\Big)^{\frac{ik}{2b}}\frac{\Ga\big(1-\frac{ik}{b}\big)\Ga\big(-\frac{ik}{b}\big)}{\Ga\big(-\frac{ik}{b}-s\big)\Ga\big(-\frac{ik}{b}+s+1\big)}F\Big(-s,s+1;1+\frac{ik}{b};\frac{1+\xi}{2}\Big),
\end{split}
\end{equation}
where $k=\sqrt{\la}$, $\xi=\tah bx$, and $s=(\sqrt{1-4\bar{m}^2c}-1)/2$. This solution passes into $e^{ikx}$ for $x\rightarrow+\infty$, i.e., this is the Jost solution $f^+(k,x)$ (see \eqref{psi12}, \eqref{psi12_asympt}). Therefore, the coefficient at $e^{ikx}$ of the asymptotics \eqref{exact_sol} for $x\rightarrow-\infty$ is exactly $s_{11}^{-1}(k)=t^{-1}(k)$ (see \eqref{psi12}). This asymptotics is readily found, and we obtain
\begin{equation}
    t(k)=\frac{\Ga\big(\al_+-\frac{ik}{b}\big)\Ga\big(\al_--\frac{ik}{b}\big)}{\Ga\big(1-\frac{ik}{b}\big)\Ga\big(-\frac{ik}{b}\big)},\qquad\al_\pm:=\frac12\pm\sqrt{\frac14-\bar{m}^2c}.
\end{equation}
The coefficient $c_1$ determining the leading asymptotics of the transition amplitude for the potential \eqref{TP_pot} is equal to
\begin{equation}
    c_1=im\bar{m}c.
\end{equation}
As a result, we have for the contribution \eqref{energy_zero_fin} to the instantaneous energy of zero-point fluctuations
\begin{equation}\label{integral}
    m\bar{m}\int_1^\infty \frac{dk}{2\pi}\sqrt{k^2-1}\Big[\psi(\bar{m}k+\al_+) +\psi(\bar{m}k+\al_-) -2\psi(\bar{m}k)-\frac{1}{\bar{m}k}-\frac{c}{k^2}\Big],
\end{equation}
where $\psi(x)=\Ga'(x)/\Ga(x)$.

In order to find the expansion of this integral at large $\bar{m}$, we employ the representation \cite{GrRy.6}
\begin{equation}\label{psi_expan}
    \psi(1+x)=-\ga+\sum_{k=2}^\infty(-1)^k\zeta(k)x^{k-1}=-\ga+\int_{\tilde{C}}\frac{dz x^{z-1}\zeta(z)}{2i\sin\pi z},
\end{equation}
where the contour $\tilde{C}$ goes from the top down along the line $\re z=3/2$. The contour $\tilde{C}$ can be displaced to the left or to the right in the interval $\re z\in(1,2)$. The last representation of $\psi(1+x)$ is valid for $\arg x\in(-\pi,\pi)$. Further, we need to substitute this representation into \eqref{integral} and evaluate the integral. At that, the Euler constants in \eqref{psi_expan} are canceled out. It is useful to make the replacement
\begin{equation}
    \sqrt{k^2-1}\rightarrow(k^2-1)^{1/2-\nu},\quad \nu\in \mathbb{C}.
\end{equation}
It is clear that, having made this replacement, the integral \eqref{integral} becomes the analytic function of $\nu$ for $\re \nu\in(-1,3/2)$. Therefore, if we find the expression for \eqref{integral} in terms of analytic function of $\nu$ in some region of the $\nu$ plane belonging to this strip and then put $\nu$ to zero, the result will be the same as if we evaluate \eqref{integral} directly. In the strip $\re\nu\in(5/4,3/2)$, we can change the order of integration over $k$ and $z$ in \eqref{integral}. Then the following integral arises
\begin{equation}
    \int_1^\infty dk(k^2-1)^{1/2-\nu}(k+\be_\pm)^{z-1}=\frac{\Ga(2\nu-z-1)\Ga(3/2-\nu)}{2^{2\nu-z-1}\Ga(\nu-z+1/2)}F(1-z,2\nu-z-1;\nu-z+1/2;(1-\be_\pm)/2),
\end{equation}
where $\be_\pm:=\bar{m}^{-1}(\al_\pm-1)$. As for the $\psi$ function in \eqref{integral} without $\al_\pm$, we have to set $\be_\pm=-\bar{m}^{-1}$. So the $\psi$ functions in \eqref{integral} give the contributions
\begin{equation}\label{integral1}
    \frac{m\bar{m}}{4^\nu\pi i}\Ga\Big(\frac32-\nu\Big)\int_{\tilde{C}}dz\frac{\Ga(1-z)\zeta(1-z)\Ga(2\nu-z-1)}{\cos(\pi z/2)}\frac{F(1-z,2\nu-z-1;\nu-z+1/2;(1-\be_\pm)/2)}{\Ga(\nu-z+1/2) (4\pi\bar{m})^{1-z}}.
\end{equation}
The integrals over $k$ of the last two terms in \eqref{integral} can be also evaluated exactly, but they will not be needed for us.

The last factor in the integrand of \eqref{integral1} is the entire function of $z$ provided $|(1-\be_\pm)/2|<1$. Inasmuch as, at $\bar{m}\rightarrow+\infty$,
\begin{equation}\label{hyper_arg}
    \frac{1-\be_\pm}2=\left\{
                        \begin{array}{ll}
                          \frac12(1\mp i\sqrt{c}); \\
                          \frac12,
                        \end{array}
                      \right.
\end{equation}
then $|(1-\be_\pm)/2|<1$ for $c\in(-1,3)$, i.e., $\bar{c}\in(-1,1)$. Consequently, the singularities of the integrand of \eqref{integral1} in the $z$ plane stem from the first factor only. In the complex half-plane $\re z\leq1$, the singular points (the poles) are located at
\begin{equation}
    z=1,0,-1,-3,-5,-7,\ldots,\qquad z=2\nu-1,2\nu,2\nu+1,2\nu+2,\ldots.
\end{equation}
We should set $\nu$ to zero in the final answer. Therefore, only three poles from the second series appear in the half-plane $\re z\leq1$.

\paragraph{Large mass expansion.}

Now we move the integration contour $\tilde{C}$ to the left. Then there will appear the contributions from the poles mentioned above. These contributions constitute the series in the decreasing powers of $\bar{m}$, the power of $\bar{m}$ being determined by the pole position. The coefficients of this series also depend on $\bar{m}$ through $\be_\pm$ standing in the argument of the hypergeometric function. The hypergeometric function is analytic in the vicinity of the points \eqref{hyper_arg}, and $(1-\be_\pm)/2$ is an analytic function of $\bar{m}^{-1}$ in the vicinity of the point $\bar{m}^{-1}=0$. So the series coefficients are expandable in a Taylor series in $\bar{m}^{-1}$ with non-zero radius of convergence. Substituting series into series, we obtain the series in the decreasing powers of $\bar{m}$ with a finite number of terms at the every power of $\bar{m}$. Moreover, as we know from the general theory described in the previous section, the large mass expansion of the integral \eqref{integral} contains only even powers of $\bar{m}$, i.e., if we take into account the contributions of the last two terms in \eqref{integral} and put $\nu=0$ then all the terms at the odd powers of $\bar{m}$ and possible terms at $\ln\bar{m}$ are canceled out. If the contour $\tilde{C}$ could be moved to the left to $-\infty$ then \eqref{integral} would be expandable in a convergent Laurent series in $\bar{m}^{-2}$. However, this cannot be done.

Indeed, let us rewrite the second factor in the integrand as
\begin{equation}
    \frac{F(1-z,2\nu-z-1;\nu-z+\frac12;\frac{1-\be_\pm}2)}{\Ga(\nu-z+\frac12) (4\pi\bar{m})^{1-z}}=\Big(\frac{1+\be_\pm}{2}\Big)^{\frac32-\nu}\frac{F(\nu-\frac12,\frac32-\nu;\nu-z+\frac12;\frac{1-\be_\pm}2)}{\Ga(\nu-z+\frac12) [2\pi\bar{m}(1+\be_\pm)]^{1-z}}.
\end{equation}
For large $\la$, $\re\la>0$, the following asymptotic expansion holds \cite{KhwDaalIV}:
\begin{equation}
    F(a,b;c+\la;x)=\sum_{s=0}^{n-1}\frac{a_sb_s}{(c+\la)_s s!}x^s+O(\la^{-n}),
\end{equation}
i.e., in the leading order,
\begin{equation}\label{hypergeom_asympt}
    F(a,b;c+\la;x)=1+O(\la^{-1}).
\end{equation}
For $\re z<-1$, we can set $\nu=0$ in \eqref{integral1}. Then, on shifting the contour $\tilde{C}$ to the region of large $|z|$, $\re z<0$, we have for \eqref{integral1} in the leading order
\begin{equation}\label{nonanalyt_contrb}
    \frac{m\bar{m}}{2i\sqrt{\pi}}\Big(\frac{1+\be_\pm}{2}\Big)^{\frac32}\int_{\tilde{C}}dz\frac{\Ga(1-z)\Ga(-1-z)}{\Ga(\frac12-z)\cos\frac{\pi z}2}[2\pi\bar{m}(1+\be_\pm)]^{z-1}.
\end{equation}
Using the asymptotic expansion of the gamma functions, we obtain
\begin{equation}\label{nonanalyt_apprx}
    \frac{m}{4i(2\pi)^{3/2}}(1+\be_\pm)^{\frac12}\int_{\tilde{C}}\frac{dz}{-z}\frac{e^{-z(\ln(-z)-1)}}{\cos(\pi z/2)}[2\pi\bar{m}(1+\be_\pm)]^{z}.
\end{equation}
As we see, the integrand grows at large negative $z$, and the contour $\tilde{C}$ cannot be moved to the left to $-\infty$.

For the contribution \eqref{integral} to the one-loop effective action to be expandable in a convergent Laurent series in $\bar{m}^{-2}$, it is necessary that
\begin{equation}\label{integral2}
    \int_{-T/2}^{T/2}dt\Big\{(1+\be_+)^{\frac12}[2\pi\bar{m}(1+\be_+)]^{z}+(1+\be_-)^{\frac12}[2\pi\bar{m}(1+\be_-)]^{z}-2(1+\be_0)^{\frac12}[2\pi\bar{m}(1+\be_0)]^{z} \Big\}=0,
\end{equation}
where $\be_0:=-\bar{m}^{-1}$. Only in this case can one move the integration contour $\tilde{C}$ to the left to $-\infty$. Generally, it may happen that the contributions from the expansion of the hypergeometric functions standing at the higher powers of $\bar{m}$ are accumulated and cancel the growth of the coefficient at $\bar{m}^{-2k}$ following from \eqref{nonanalyt_apprx}. However, it follows from the asymptotics \eqref{hypergeom_asympt} that this does not occur. The integral \eqref{integral2}, in the leading order in $\bar{m}$, can be cast into the form
\begin{equation}\label{integral3}
    \int_{0}^{T/2}dt(2\pi\bar{m})^z\Big[(1+i\sqrt{c})^{z+\frac12} +(1-i\sqrt{c})^{z+\frac12} +(1+\sqrt{c})^{z+\frac12} +(1-\sqrt{c})^{z+\frac12} -4\Big],
\end{equation}
where we have taken into account the evenness of $b(t)$ and the oddness of $c(t)$ at small $c(t)$. For $c(t)$ small, the integral \eqref{integral3} can be also approximately written as
\begin{equation}
    2\int_{0}^{T/2}dt(2\pi\bar{m})^z\Big\{\cos\big[(z+1/2)\sqrt{c}\big]+\ch\big[(z+1/2)\sqrt{c}\big] -2\Big\}.
\end{equation}
It is clear that neither this expression nor \eqref{integral3} are equal to zero for the arbitrary functions $b(t)$ and $c(t)$ satisfying \eqref{exampl_phi}, \eqref{ccbar}.

Let us find the explicit expression for the non-analytic in $\bar{m}^{-2}$ terms \eqref{nonanalyt_contrb} in \eqref{integral} in the leading order in $\bar{m}\rightarrow+\infty$ assuming, for simplicity, that $b(t)=const$. To this aim, we shift the contour $\tilde{C}$ in \eqref{nonanalyt_contrb} to the left until then the modulus of the integrand (apart from $\cos^{-1}(\pi z/2)$) reaches the minimum. This occurs at the different positions of the contour $\tilde{C}$ for the different $\psi$ functions entering \eqref{integral}. As seen from \eqref{nonanalyt_apprx}, for $c>0$, the nearest to the $z$ plane origin extremum point is situated approximately at
\begin{equation}\label{c_pos}
    z_{ext}=-2\pi(\bar{m}-1),
\end{equation}
and for $c<0$ at
\begin{equation}\label{c_neg}
    z_{ext}=-2\pi\bar{m}(1-\sqrt{-c})+\pi.
\end{equation}
Let us choose the contour $\tilde{C}$ be the same for all the $\psi$ functions. In the case $c>0$, it passes near \eqref{c_pos} and, in the case $c<0$, it goes near \eqref{c_neg} with $c\rightarrow c_0:=\min c(t)$.

Let $p\sim-\bar{m}$. We parameterize the contour $\tilde{C}$ as
\begin{equation}
    z=p+ix,\qquad x\in \mathbb{R}.
\end{equation}
Then, developing the integrand of \eqref{nonanalyt_contrb} (apart from $\cos^{-1}(\pi z/2)$) as a series in $x$ in the neighbourhood of the point $z=p$, we obtain for non-vanishing at $p\rightarrow -\infty$ terms
\begin{equation}\label{expns}
    \frac{\sqrt{2\pi}}{ap}e^{-(p+\ln(a/p))+ix\ln(a/p)},
\end{equation}
where the notation $a:=2\pi\bar{m}(1+\be_\pm)$ has been introduced, and it has been taken into account that, in integrating over $x$, the order of $|x|\lesssim1$ due to the ``cut-off'' factor $\cos^{-1}(\pi z/2)$.

For $c>0$, we put $p=-2n$, $n=[\pi\bar{m}]$. Then
\begin{equation}
    \cos^{-1}(\pi z/2)=(-1)^n\ch^{-1}(\pi x/2).
\end{equation}
Substituting the expansion \eqref{expns} into \eqref{nonanalyt_contrb} and integrating over $x$, we derive that the non-analytic contribution of one $\psi$ function to \eqref{integral} is given by
\begin{equation}\label{nonanalyt_int}
    m\bar{m}(-1)^ne^{-2n}\frac{(1+\be_\pm)^{3/2}}{a^2+p^2}\frac{p}{a}.
\end{equation}
Collecting the contributions of all the $\psi$ functions, we find for the non-analytic exponentially suppressed contributions to \eqref{integral}
\begin{equation}
    (-1)^n\frac{be^{-2n}}{4\pi^2}\Big[\frac{(1+i\sqrt{c})^{1/2}}{1+(1+i\sqrt{c})^2} +\frac{(1-i\sqrt{c})^{1/2}}{1+(1-i\sqrt{c})^2} -1\Big],
\end{equation}
in the limit $\bar{m}\rightarrow+\infty$. For $c<0$, we put $p=-2n$, $n=[\pi\bar{m}(1-\sqrt{-c_0})]$. Using \eqref{nonanalyt_int}, we obtain in this case that the non-analytic exponentially suppressed contributions to \eqref{integral} are
\begin{equation}
    (-1)^n\frac{be^{-2n}}{4\pi^2}(1-\sqrt{-c_0})\Big[\frac{(1+\sqrt{-c})^{1/2}}{(1+\sqrt{-c})^{2}+(1-\sqrt{-c_0})^2} +\frac{(1-\sqrt{-c})^{1/2}}{(1-\sqrt{-c})^{2}+(1-\sqrt{-c_0})^2} -\frac{2}{1+(1-\sqrt{-c_0})^2}\Big],
\end{equation}
in the leading order at $\bar{m}\rightarrow+\infty$.

The derived non-analytic in $m^{-2}$ contributions to the one-loop effective action are proportional to
\begin{equation}\label{gen_struct}
    e^{-km/b}.
\end{equation}
As shown in \cite{KalKaz2.12}, the quantity making dimensionless the mass in the exponent cannot be a scalar constructed in terms of the metric, its curvature, and the covariant derivatives of the curvature. It is evident in our case since the metric curvature is zero. It is clear that this quantity cannot contain the non-local contributions of the form $(\nabla^2+m^2)^{-1}$ etc. since, in the large mass limit we consider, these contributions become local. It was shown in \cite{KalKaz2.12}, in a certain approximation, that the quantity making dimensionless the mass in the exponent is a scalar made of the metric and the covariant derivatives of the time-like vector field defining the Hamiltonian. The contributions \eqref{gen_struct} cannot be canceled out by the counterterms polynomial in momenta, the fields, and the coupling constant\footnote{Some necessary facts from renormalization theory are collected in App. \ref{Reg_Renorm}.}. In the one-loop approximation, these contributions are not canceled out by the contributions of the ghosts and gravitons (in the four-dimensional spacetime) due to the fact that the background propagators of the ghosts and gravitons do not depend on the mass of a scalar field. We shall discuss some more this result in Conclusion.

\paragraph{Method of comparison equations.}

Let us estimate, in the semiclassical approximation, the correction to the transition amplitude $t(k)$ resulting from the potential difference $V-V_0$. To this end, we employ the so-called method of comparison equations (see the review \cite{BerrMount}, Sect. 4). According to this method, a good approximation to the solution of the Schr\"{o}dinger equation with the potential $V$ is
\begin{equation}\label{appr_sol}
    u(\la,x)=\Big[\frac{\la-V_0(\s(x))}{\la-V(x)}\Big]^{1/4}\psi(\la,\s(x)),
\end{equation}
where $\psi(\la,\s)$ is the exact solution of the Schr\"{o}dinger equation with the potential $V_0$ and the eigenvalue $\la$, i.e., \eqref{exact_sol} in our case. The function $\s(x)$ is determined by the equation
\begin{equation}\label{sigma_x}
    \int_0^{\s}dy\sqrt{\tilde{\la}-V_0(y)}=\int_0^{x}dy\sqrt{\tilde{\la}-V(y)},
\end{equation}
where $\tilde{\la}:=\la/m^2$. The constant $c$ entering the potential $V_0$ is found from the turning point matching condition
\begin{equation}\label{concord}
    \int_0^{\s_3}dy\sqrt{\tilde{\la}-V_0(y)}=\int_0^{x_3}dy\sqrt{\tilde{\la}-V(y)},
\end{equation}
where $\s_3$ and $x_3$ are the turning points for the potentials $V_0$ and $V$, respectively. In general, $c=c(\tilde{\la},\bar{c})$.

To apply this method exactly is apparently impossible in our case. However, if one takes into account that, at the small $c$ and $\bar{c}$, the difference of the potentials $V_0$ and $V$ is small, \eqref{concord} can be approximately written as
\begin{equation}
    \int_0^{\s_3}dy\sqrt{\tilde{\la}-V_0(y)}=\int_0^{x_3}dy\sqrt{\tilde{\la}-V(y)}\approx \int_0^{\s_3}dy\sqrt{\tilde{\la}-V_0(y)}- \int_0^{\s_3}\frac{dy\de V(y)}{2\sqrt{\tilde{\la}-V_0(y)}}.
\end{equation}
The last term has the form
\begin{equation}
    -\int_0^{\s_3}\frac{dy}{2\sqrt{\tilde{\la}-V_0(y)}}\Big[\frac{2\bar{c}-c}{\ch^2(bx)}+\frac{\bar{c}^2}{\ch^4(bx)}\Big].
\end{equation}
Thus, in the leading non-trivial order in small $c$, the turning point matching condition leads to
\begin{equation}
    c=2\bar{c},
\end{equation}
i.e. in fact, to the condition \eqref{ccbar}.

In order to find the asymptotics of the approximate solution \eqref{appr_sol}, it is necessary to find the asymptotic behavior of $\s(x)$ at $x\rightarrow\pm\infty$. Writing \eqref{sigma_x} at large $\s$ and $x$ as
\begin{equation}
    \int_{\s_3}^\s dy\sqrt{V_0(y)-\tilde{\la}}=\int_{x_3}^xdy\sqrt{V(y)-\tilde{\la}}\approx\int_{\s_3}^x dy\sqrt{V_0(y)-\tilde{\la}}+\int_{\s_3}^x\frac{dy\de V(y)}{2\sqrt{V_0(y)-\tilde{\la}}},
\end{equation}
we come to
\begin{equation}
    \int_x^\s dy\sqrt{V_0(y)-\tilde{\la}}\approx\int_{\s_3}^x \frac{dy\de V(y)}{2\sqrt{V_0(y)-\tilde{\la}}}.
\end{equation}
For $x$ and $\s$ tending to $+\infty$, we have
\begin{equation}
    b(\s-x)\sqrt{-\tilde{\la}}\approx\frac{i}{2}\de I_{12}=i\sqrt{\tilde{\la}} \frac{\bar{c}^2}{4c}\Big[\Big(\sqrt{\frac{c}{\tilde{\la}}}-\sqrt{\frac{\tilde{\la}}{c}}\Big)\arth\sqrt{\frac{c}{\tilde{\la}}} +1\Big],
\end{equation}
where the dependence on $b$ is restored, and recall, the principal branches of the multivalued functions are taken. Consequently, for $x$ and $\s$ tending to $+\infty$ and $-\infty$, we obtain
\begin{equation}
    \s=x+\frac{\bar{c}^2}{4bc}\Big[\Big(\sqrt{\frac{c}{\tilde{\la}}}-\sqrt{\frac{\tilde{\la}}{c}}\Big)\arth\sqrt{\frac{c}{\tilde{\la}}} +1\Big],\qquad
    \s=x-\frac{\bar{c}^2}{4bc}\Big[\Big(\sqrt{\frac{c}{\tilde{\la}}}-\sqrt{\frac{\tilde{\la}}{c}}\Big)\arth\sqrt{\frac{c}{\tilde{\la}}} +1\Big],
\end{equation}
respectively. Now it is not difficult to see that the use of this method results in the correction to $\ln t(k)$ of the form
\begin{equation}
    \ln t_{corr}(k)=i\frac{k\bar{c}^2}{2bc}\Big[\Big(\frac{m\sqrt{-c}}{-ik}-\frac{-ik}{m\sqrt{-c}}\Big)\arth\frac{m\sqrt{-c}}{-ik} +1\Big].
\end{equation}
So we have
\begin{equation}
    ic_1^{corr}=\frac{\bar{c}^2}{3}m\bar{m},
\end{equation}
and the correction to the contribution \eqref{integral}
\begin{equation}
    -m\bar{m}\frac{\bar{c}^2}{c}\int_1^\infty dk\sqrt{k^2-1}\Big[1-\frac{k}{\sqrt{-c}}\arcth\frac{k}{\sqrt{-c}}+\frac{c}{3k^2}\Big].
\end{equation}
As a result, we see that the difference of the potentials $V$ and $V_0$ just leads to small corrections to the conformal anomaly and the coefficient at $m^2$ in the leading non-trivial order in $c$. The presence of the correction to the conformal anomaly ought to be expected since $V\neq V_0$.

\section{Conclusion}

The main result of the present paper can be formulated as a theorem.

\begin{thm}
  The real part of the one-loop effective action, in the background field gauge, induced by a massive scalar field on the two-dimensional Minkowski spacetime depends on the choice of the smooth splitting of the spacetime onto the space and time (i.e., on the choice of the time coordinate) for $t\in(t_{in},t_{out})$, and this dependence cannot be canceled out by the polynomial in derivatives and fields counterterms added to the initial classical action.
\end{thm}

A similar statement for a stationary slowly varying in space metric background was formulated and proved in \cite{KalKaz2.12}, but this proof involves many approximations the validity of which is rather hard to control on the mathematical level of rigor. As shown in \cite{KalKaz2.12} and it is obvious from the cause of dependence of the effective action on the splitting (see below), this statement holds for the four-dimensional spacetime as well. Therefore, the one-loop effective action induced by the massive scalar field, gravitons, and ghosts, in the background field gauge, is not invariant under diffeomorphisms acting on the dynamical fields of the theory on the solutions to the equations of motion of the fields $\phi$ and ghosts. The diffeomorphisms are supposed to be identical at the spatial and temporal infinities.

To avoid misunderstanding, let us stress that we distinguish the notions of a diffeomorphism and a change of coordinates. A change of coordinates is the transition from one chart to another in the bundle, while a diffeomorphism moves the manifold points and transports the tensor fields appropriately. For example, the action functional
\begin{equation}\label{action_ex}
    S[g]=\int d^4x\sqrt{|g|} (g^{\mu\nu}\partial_\mu\ln\sqrt{\xi^2}\partial_\nu\ln\sqrt{\xi^2})^2,
\end{equation}
where $\xi^\mu$ is some non-dynamical field and $\xi^2=g_{\mu\nu}\xi^\mu\xi^\nu$, is invariant under a change of coordinates. Nevertheless, \eqref{action_ex} is not invariant under diffeomorphisms acting on the dynamical fields of the theory, i.e., this action is not gauge invariant\footnote{The gauge transformations act solely on the dynamical fields of the theory, i.e., on the fields with respect to which the action is minimized. This becomes quite evident, for example, in the Hamiltonian formalism.}. Notice, in this respect, a certain peculiarity that sometimes can lead astray. If $g_{\mu\nu}$ is a stationary metric and $\xi^\mu$ is the Killing vector then the energy-momentum tensor resulting from a variation of \eqref{action_ex} with respect to $g_{\mu\nu}$ at fixed $\xi^\mu$ is covariantly divergenceless \cite{FrZel.11}. This, however, does not mean that \eqref{action_ex} is generally covariant since the higher Noether identities ensuing from a variation of the energy-momentum tensor divergence with respect to the metric do not hold.

These reasonings on the invariance of a theory with respect to diffeomorphisms are immediately extended to the quantum effective action in the background field gauge \cite{DeWGAQFT.11,BuchOdinShap.11,WeinB2}. Then the Noether identities become the Ward ones:
\begin{equation}\label{ward_id_expct}
    \e^\al R_\al^i(\Phi)\frac{\de\Ga}{\de\Phi^i}\equiv0,
\end{equation}
where $\Phi$ is a complete set of the background fields and $R^i_\al$ are the gauge transformation generators. Expanding \eqref{ward_id_expct} in $\hbar$, we obtain, in the leading non-trivial order in $\hbar$,
\begin{equation}\label{grav_anom}
    \e^\al R_\al^i(\Phi)\frac{\de\Ga}{\de\Phi^i}=\hbar\e^\al R_{(0)\al}^i(\Phi)\frac{\de\Ga_{(1)}}{\de\Phi^i}+\cdots,
\end{equation}
on the solutions to the classical equations of motion. The Minkowski metric is evidently the solution to the Einstein equations without the cosmological constant. Therefore, the statement proved by us is equivalent to the statement on the presence of the one-loop non-perturbative quantum gravitational anomaly in the theory of the massive scalar field and gravitons in the Minkowski spacetime. The cause for the appearance of this anomaly in relativistic QFT is the necessity to split the field operators onto positive- and negative-frequency parts (see the discussion of this problem in \cite{DeWpaper.9}). As we saw in Sect. \ref{Hamilt_Diag}, such a natural polarization is given by the Hamiltonian of a theory. To construct the Hamiltonian, it is necessary to define a certain time-like vector field $\partial_t$, the contractions and covariant derivatives of which arise eventually in the one-loop effective action. A possible mechanism allowing to get rid of this anomaly is to endow the field $\xi^\mu\partial_\mu=\partial_t$ with its own dynamics. Then the effective action will be invariant under the action of diffeomorphisms on the equations of motion of the fields $\phi$, $\xi^\mu$, and ghosts. One can make the field $\xi^\mu$ dynamical in an infinity of ways. In a general case, $\xi^\mu$ is not a fundamental field of a theory but a local composite operator of the fields $\phi^a$, which realize a certain representation of the diffeomorphism group and such that $\xi^\mu$ transforms as a vector field under the action of diffeomorphisms. A thorough discussion of the most natural mechanism for cancelation of the quantum gravitational anomaly considered and the possible physical consequences of this mechanism can be found in \cite{qgadm}.

\appendix
\section{Evolution operator symbol}\label{Second_Quant}

For the reader convenience, we provide in this appendix some formulas and theorems from \cite{BerezMSQ1.4} without proofs. The proofs of these statements can be found in \cite{BerezMSQ1.4}.

Let
\begin{equation}\label{state_fock}
    |\Phi\ran=K_{\al_1\cdots\al_n}\frac{\hat{a}^\dag_{\al_1}\cdots \hat{a}^\dag_{\al_n}}{\sqrt{n!}}|0\ran
\end{equation}
be a state in the Fock space. The summation over all the repeating indices (including $n$) is assumed.
\begin{defn}
  The generating functional of the state \eqref{state_fock} is the functional
  \begin{equation}
    \Phi(\bar{a}):=K_{\al_1\cdots\al_n}\frac{\bar{a}_{\al_1}\cdots \bar{a}_{\al_n}}{\sqrt{n!}}.
  \end{equation}
\end{defn}

Let
\begin{equation}\label{oper_fock}
    \hat{A}=K_{\al_1\cdots\al_n;\be_1\cdots\be_m}\hat{a}^\dag_{\al_1}\cdots \hat{a}^\dag_{\al_n}\hat{a}_{\be_1}\cdots \hat{a}_{\be_m}
\end{equation}
be a normal ordered operator acting in the Fock space.
\begin{defn}
  The Wick symbol of the operator \eqref{oper_fock} is called the functional
  \begin{equation}
    A(\bar{a},a):=K_{\al_1\cdots\al_n;\be_1\cdots\be_m}\bar{a}_{\al_1}\cdots \bar{a}_{\al_n}a_{\be_1}\cdots a_{\be_m}.
  \end{equation}
\end{defn}

Let us introduce the notation for the states in the Fock space
\begin{equation}
    |\al_1,\cdots,\al_n\ran:=\frac{\hat{a}^\dag_{\al_1}\cdots \hat{a}^\dag_{\al_n}}{\sqrt{n!}}|0\ran,\qquad \lan \al_1,\cdots,\al_n|:=\lan0|\frac{\hat{a}_{\al_1}\cdots \hat{a}_{\al_n}}{\sqrt{n!}}.
\end{equation}
These states are not normalized to unity.

\begin{defn}
  The generating functional for the matrix form of the operator \eqref{oper_fock} is called the functional
  \begin{equation}
    \tilde{A}(\bar{a},a):=\lan\al_1,\cdots,\al_n|\hat{A}|\be_1,\cdots,\be_m\ran\frac{\bar{a}_{\al_n}\cdots \bar{a}_{\al_1}a_{\be_m}\cdots a_{\be_1}}{\sqrt{n!m!}}.
  \end{equation}
\end{defn}

The following relations hold
\begin{equation}\label{rels_symb}
\begin{alignedat}{4}
    &&\hat{a}_\al|\Phi\ran &\leftrightarrow\frac{\de\Phi(\bar{a})}{\de\bar{a}_\al},&\quad \hat{a}^\dag_\al|\Phi\ran &\leftrightarrow \bar{a}_\al\Phi(\bar{a}),&&\\
    \hat{a}_\al\hat{A} &\leftrightarrow \Big(a_\al+\frac{\de}{\de\bar{a}_\al}\Big)A(\bar{a},a),&\quad \hat{a}^\dag_\al\hat{A}&\leftrightarrow \bar{a}_\al A(\bar{a},a),&\quad \hat{A}\hat{a}_\al &\leftrightarrow A(\bar{a},a)a_\al,&\quad \hat{A}\hat{a}^\dag_\al &\leftrightarrow A(\bar{a},a)\Big(\bar{a}_\al+\frac{\overleftarrow{\de}}{\de a_\al}\Big),\\
    \hat{a}_\al\hat{A} &\leftrightarrow \frac{\de \tilde{A}(\bar{a},a)}{\de\bar{a}_\al},&\quad \hat{a}^\dag_\al\hat{A} &\leftrightarrow \bar{a}_\al \tilde{A}(\bar{a},a),&\quad \hat{A}\hat{a}_\al &\leftrightarrow \tilde{A}(\bar{a},a)a_\al,&\quad \hat{A}\hat{a}^\dag_\al &\leftrightarrow \tilde{A}(\bar{a},a)\frac{\overleftarrow{\de}}{\de a_\al}.
\end{alignedat}
\end{equation}
The Wick symbol and the generating functional for the matrix form of the operator are related as
\begin{equation}
   \tilde{A}(\bar{a},a)=A(\bar{a},a)e^{\bar{a}_\al a_\al}.
\end{equation}

Further, we consider the bosonic case that we are interested in. Let the two sets of the creation-annihilation operators $(\hat{a}_\al,\hat{a}^\dag_\al)$ and $(\hat{b}_\al,\hat{b}^\dag_\al)$ be related by a linear canonical transform
\begin{equation}\label{canon_trans}
    \left[
       \begin{array}{c}
         \hat{b} \\
         \hat{b}^\dag \\
       \end{array}
     \right]=\left[
               \begin{array}{cc}
                 \Phi & \Psi \\
                 \bar{\Psi} & \bar{\Phi} \\
               \end{array}
             \right]\left[
                      \begin{array}{c}
                        \hat{a} \\
                        \hat{a}^\dag \\
                      \end{array}
                    \right]+\left[
                              \begin{array}{c}
                                f \\
                                \bar{f} \\
                              \end{array}
                            \right],
\end{equation}
where
\begin{equation}
    \left[
       \begin{array}{cc}
         \Phi & \Psi \\
         \bar{\Psi} & \bar{\Phi} \\
       \end{array}
     \right]\left[
              \begin{array}{cc}
                0 & 1 \\
                -1 & 0 \\
              \end{array}
            \right]\left[
                     \begin{array}{cc}
                       \Phi^T & \Psi^\dag \\
                       \Psi^T & \Phi^\dag \\
                     \end{array}
                   \right]=\left[
              \begin{array}{cc}
                0 & 1 \\
                -1 & 0 \\
              \end{array}
            \right],
\end{equation}
or
\begin{equation}\label{canon_defn}
    \left[
       \begin{array}{cc}
         \Phi & \Psi \\
         \bar{\Psi} & \bar{\Phi} \\
       \end{array}
     \right]
     \left[
       \begin{array}{cc}
         \Phi^\dag & -\Psi^T \\
         -\Psi^\dag & \Phi^T \\
       \end{array}
     \right]=
     \left[
       \begin{array}{cc}
         \Phi^\dag & -\Psi^T \\
         -\Psi^\dag & \Phi^T \\
       \end{array}
     \right]
     \left[
       \begin{array}{cc}
         \Phi & \Psi \\
         \bar{\Psi} & \bar{\Phi} \\
       \end{array}
     \right]=
     \left[
       \begin{array}{cc}
         1 & 0 \\
         0 & 1 \\
       \end{array}
     \right].
\end{equation}

\begin{thm}\label{canon_unit_thm}
  The linear canonical transform \eqref{canon_trans}, \eqref{canon_defn} corresponds to the unitary transform
  \begin{equation}\label{unit_trans_2}
    \hat{b}_\al=\hat{U}\hat{a}_\al\hat{U}^\dag,\qquad \hat{b}^\dag_\al=\hat{U}\hat{a}^\dag_\al\hat{U}^\dag,
  \end{equation}
  if and only if
  \begin{enumerate}
    \item $\Psi$ is a Hilbert-Schmidt operator
    \begin{equation}
        \Sp(\Psi^\dag\Psi)<\infty;
    \end{equation}
    \item $f_\al$ belongs to the Hilbert space
\begin{equation}
    \sum_\al\bar{f}_\al f_\al<\infty.
\end{equation}
  \end{enumerate}
  In this case, the generating functional $\tilde{U}(\bar{a},a)$ for the matrix form of the operator $\hat{U}$ is
  \begin{equation}\label{unit_trans}
  \begin{split}
    \tilde{U}=&c\exp\Big\{\frac12\left[
                                  \begin{array}{cc}
                                    a & \bar{a} \\
                                  \end{array}
                                \right]
    \left[
                                             \begin{array}{cc}
                                               \bar{\Psi}\Phi^{-1} & (\Phi^{-1})^T \\
                                               \Phi^{-1} & -\Phi^{-1}\Psi \\
                                             \end{array}
                                           \right]\left[
                                                    \begin{array}{c}
                                                      a \\
                                                      \bar{a} \\
                                                    \end{array}
                                                  \right]+a(\bar{f}-\bar{\Psi}\Phi^{-1}f)-\bar{a}\Phi^{-1}f
     \Big\},\\
     c=&\frac{e^{i\vf}}{(\det\Phi\Phi^\dag)^{1/4}}\exp\Big\{\frac14\left[
                                                                     \begin{array}{cc}
                                                                       f & \bar{f} \\
                                                                     \end{array}
                                                                   \right]
                                                                   \left[
                                                                     \begin{array}{cc}
                                                                       (\Phi^{-1})^T\Psi^\dag & -1 \\
                                                                       -1 &  (\Phi^{-1})^\dag\Psi^T\\
                                                                     \end{array}
                                                                   \right]
                                                                   \left[
                                                                     \begin{array}{c}
                                                                       f \\
                                                                       \bar{f} \\
                                                                     \end{array}
                                                                   \right]
      \Big\},
  \end{split}
  \end{equation}
  where $\vf$ is an arbitrary real phase.
\end{thm}

Notice that the Fredholm determinant appearing in \eqref{unit_trans} exists since $(\Phi\Phi^\dag-1)$ is a trace-class operator.

Let the Hamilton operator of a system be reduced to the normal form
\begin{equation}\label{Hamil_gener}
    \hat{H}=\frac12\big[2\hat{a}^\dag_\al C_{\al\be}(t)\hat{a}_\al+\hat{a}_\al \bar{A}_{\al\be}(t) \hat{a}_\be+\hat{a}^\dag_\al A_{\al\be}(t) \hat{a}^\dag_\be\big]+\hat{a}^\dag_\al f_\al(t)+\bar{f}_\al(t)\hat{a}_\al,
\end{equation}
where $C(t)=C^\dag(t)$ is a self-adjoint operator, $A(t)=A^T(t)$ is a symmetric Hilbert-Schmidt operator, and $f_\al(t)$ belongs to the Hilbert space.

\begin{thm}\label{evol_symb_thm}
  Suppose that the Hamilton operator has the form \eqref{Hamil_gener} and the following conditions hold:
  \begin{enumerate}
    \item There exists the unitary operator\footnote{The sufficient conditions for the existence of such an operator can be found in \cite{Howland} and \cite{ReedSimon2}, Sect. 13.}
    \begin{equation}\label{V_oper}
        V(t)=\Texp\Big[-i\int_0^t d\tau \bar{C}(\tau)\Big];
    \end{equation}
    \item The operator
    \begin{equation}
        F(t):=\int_0^td\tau V^T(\tau)A(\tau)V(\tau)
    \end{equation}
    is Hilbert-Schmidt and the function
    \begin{equation}
        \Sp[F(t)F^\dag(t)]^{1/2}
    \end{equation}
    is locally summable;
    \item The operator $F(t)\bar{A}(t)$ is trace-class and the function
    \begin{equation}
        \Sp\big\{F(t)\bar{A}(t)[F(t)\bar{A}(t)]^\dag\big\}^{1/2}
    \end{equation}
    is locally summable.
  \end{enumerate}
  Then the generating functional for the matrix form of the unitary evolution operator $\hat{U}_{t,0}$ is
  \begin{equation}\label{evol_symb}
    \tilde{U}_{t,0}(\bar{a},a)=c(t)\exp\Big\{\frac12\left[
                                                      \begin{array}{cc}
                                                        \bar{a} & a \\
                                                      \end{array}
                                                    \right]
                                                    \left[
                                                      \begin{array}{cc}
                                                        \Psi\bar{\Phi}^{-1} & (\Phi^{-1})^\dag \\
                                                        \bar{\Phi}^{-1} & -\bar{\Phi}^{-1}\bar{\Psi} \\
                                                      \end{array}
                                                    \right]
                                                    \left[
                                                      \begin{array}{c}
                                                        \bar{a} \\
                                                        a \\
                                                      \end{array}
                                                    \right]
                                                    +\bar{a}(g-\Psi\bar{\Phi}^{-1}\bar{g})-a\bar{\Phi}^{-1}\bar{g}
    \Big\},
  \end{equation}
  where
  \begin{equation}\label{gandc}
  \begin{split}
    \left[
      \begin{array}{c}
        g \\
        \bar{g} \\
      \end{array}
    \right]&=-i\int_0^t d\tau\Texp\Big\{-i\int_\tau^tds\left[
                                           \begin{array}{cc}
                                             C(s) & A(s) \\
                                             -\bar{A}(s) & -\bar{C}(s) \\
                                           \end{array}
                                         \right]\Big\}
                                         \left[
                                           \begin{array}{c}
                                             f(\tau) \\
                                             -\bar{f}(\tau) \\
                                           \end{array}
                                         \right],\\
    c(t)&=\big[\det\bar{\Phi}(t)V(t)\big]^{-1/2}\exp\Big\{-i\int_0^td\tau\Big[\frac12(g-\bar{g}(\Phi^{-1})^\dag\Psi^T)\bar{A}(g-\Psi\bar{\Phi}^{-1}\bar{g})+\bar{f}(g-\Psi\bar{\Phi}^{-1}\bar{g})\Big] \Big\},
  \end{split}
  \end{equation}
and also
\begin{equation}\label{unit_canon}
    \left[
      \begin{array}{cc}
        \Phi(t) & \Psi(t) \\
        \bar{\Psi}(t) & \bar{\Phi}(t) \\
      \end{array}
    \right]
    =\Texp\Big\{-i\int_0^t d\tau \left[
                                           \begin{array}{cc}
                                             C(\tau) & A(\tau) \\
                                             -\bar{A}(\tau) & -\bar{C}(\tau) \\
                                           \end{array}
                                         \right]\Big\}.
\end{equation}
\end{thm}
\begin{proof}
  We shall not present here a rigorous proof of this theorem. It is given in \cite{BerezMSQ1.4}. We provide only the formal calculations showing the correctness of formula \eqref{evol_symb}. The conditions of the theorem guarantee the existence of the evolution operator, of the operator \eqref{unit_canon}, and of the Fredholm determinant in $c(t)$.

  Let us introduce the creation-annihilation operators related to $(\hat{a}_\al(0),\hat{a}^\dag_\al(0))$ by a unitary transform,
  \begin{equation}\label{unit_trans_3}
    \hat{a}_\al(t):=\hat{U}_{0,t}\hat{a}_\al(0)\hat{U}_{t,0},\qquad \hat{a}^\dag_\al(t):=\hat{U}_{0,t}\hat{a}^\dag_\al(0)\hat{U}_{t,0}.
  \end{equation}
  These operators obey the Heisenberg equations with the Hamilton operator \eqref{Hamil_gener}
  \begin{equation}
    i\left[
       \begin{array}{c}
         \dot{\hat{a}}(t) \\
         \dot{\hat{a}}^\dag(t) \\
       \end{array}
     \right]=
     \left[
       \begin{array}{cc}
         C(t) & A(t) \\
         -\bar{A}(t) & -\bar{C}(t) \\
       \end{array}
     \right]
     \left[
       \begin{array}{c}
         \hat{a}(t) \\
         \hat{a}^\dag(t) \\
       \end{array}
     \right]+
     \left[
       \begin{array}{c}
         f(t) \\
         -\bar{f}(t) \\
       \end{array}
     \right].
  \end{equation}
  The solution to these equations is written as
  \begin{equation}
    \left[
      \begin{array}{c}
        \hat{a}(t) \\
        \hat{a}^\dag(t) \\
      \end{array}
    \right]=
    \left[
      \begin{array}{cc}
        \Phi(t) & \Psi(t) \\
        \bar{\Psi}(t) & \bar{\Phi}(t) \\
      \end{array}
    \right]
    \left[
      \begin{array}{c}
        \hat{a}(0) \\
        \hat{a}^\dag(0) \\
      \end{array}
    \right]
    +\left[
      \begin{array}{c}
        g \\
        \bar{g} \\
      \end{array}
    \right],
  \end{equation}
  where the expressions on the right-hand side are given in formulas \eqref{gandc}, \eqref{unit_canon}. Since the creation-annihilation operators $(\hat{a}(t),\hat{a}^\dag(t))$ and $(\hat{a}(0),\hat{a}^\dag(0))$ are related by a unitary transform, we can employ the theorem \ref{canon_unit_thm} to construct the explicit expression for the generating functional of the matrix form of the evolution operator (up to a phase). Comparing \eqref{unit_trans_2} with \eqref{unit_trans_3}, we obtain the formula \eqref{evol_symb} with the factor $c(t)$ independent from $(\bar{a},a)$. This factor is to be found directly from the Schr\"{o}dinger equation for the evolution operator. To deduce \eqref{evol_symb} from \eqref{unit_trans}, it is necessary to remember that for the Hermitian conjugation of an operator there corresponds the complex conjugation of the generating functional of its matrix form.

  Taking into account the relations \eqref{rels_symb}, we obtain for the Schr\"{o}dinger equation
  \begin{equation}
    i\partial_t\tilde{U}_{t,0}=\bigg[\frac12\Big(2\bar{a}C(t)\frac{\de}{\de\bar{a}}+\bar{a}A(t)\bar{a}+\frac{\de}{\de\bar{a}}\bar{A}(t)\frac{\de}{\de\bar{a}}\Big) +f(t)\bar{a}+\bar{f}(t)\frac{\de}{\de\bar{a}} \bigg]\tilde{U}_{t,0}.
  \end{equation}
  Substituting the complex conjugate expression \eqref{unit_trans} into this equation and putting $a=\bar{a}=0$, we arrive at
  \begin{equation}
    i\dot{c}=\Big[\frac12\Sp(\Psi\bar{\Phi}^{-1}\bar{A})+\frac12(g-\bar{g}(\Phi^{-1})^\dag\Psi^T)\bar{A}(g-\Psi\bar{\Phi}^{-1}\bar{g}) +\bar{f}(g-\Psi\bar{\Phi}^{-1}\bar{g})\Big],\qquad c(0)=1.
  \end{equation}
  It is not difficult to check that $c(t)$ presented in \eqref{gandc} is the solution to this equation.
\end{proof}

The matrix element of the evolution operator,
\begin{equation}\label{vac_vac}
    \lan0|\hat{U}_{t,0}|0\ran=U_{t,0}(0,0)=\tilde{U}_{t,0}(0,0)=c(t),
\end{equation}
gives the vacuum-vacuum transition amplitude for a system with the Hamiltonian \eqref{Hamil_gener}. If $f_\al(t)=0$ then
\begin{equation}\label{vac_vac_expl}
    c(t)=\big[\det\bar{\Phi}(t)V(t)\big]^{-1/2}.
\end{equation}
Under the assumption that the operator $A$ is small in a certain sense, one can develop the perturbation theory to calculate the Fredholm determinant considered. At first, we observe that for $A=0$ this determinant is equal to unity. Let
\begin{equation}
    h=\left[
       \begin{array}{cc}
         C(t) & A(t) \\
         -\bar{A}(t) & -\bar{C}(t) \\
       \end{array}
     \right]=\left[
       \begin{array}{cc}
         C(t) & 0 \\
         0 & -\bar{C}(t) \\
       \end{array}
     \right]+
     \left[
       \begin{array}{cc}
         0 & A(t) \\
         -\bar{A}(t) & 0 \\
       \end{array}
     \right]=h_0(t)+v(t),
\end{equation}
where $h_0(t)$ is the first matrix and $v(t)$ is the second one. Let us introduce the operator
\begin{equation}
    S_{t,0}:=U^0_{0,t}U_{t,0},\qquad U^0_{t_2,t_1}=:
    \left[
      \begin{array}{cc}
        R_{t_2,t_1} & 0 \\
        0 & \bar{R}_{t_2,t_1} \\
      \end{array}
    \right],
\end{equation}
where $U^0_{t,0}$ is the evolution operator associated with $h_0(t)$ and $U_{t,0}$ is the evolution operator constructed with the aid of $h(t)$. There is the standard representation for the operator $S_{t,0}$:
\begin{equation}
    S_{t,0}=\Texp\Big[-i\int_0^t d\tau v_I(\tau)\Big],\qquad v_I(t)=U^0_{0,t}v(t)U^0_{t,0}.
\end{equation}
The determinant of the block $(1,1)$ of the matrix $S_{t,0}$ is the determinant entering into \eqref{gandc}. As a result, we obtain
\begin{equation}
    \ln\det[\Phi(t)\bar{V}(t)]=\Sp\ln\Big[1+\sum_{n=1}^\infty\frac{(-i)^n}{n!}\int_0^t d\tau_1\cdots d\tau_n P\, T\{v_I(\tau_1)\cdots v_I(\tau_n)\}P\Big],\quad P=\left[
                                                                                                                                                                 \begin{array}{cc}
                                                                                                                                                                   1 & 0 \\
                                                                                                                                                                   0 & 0 \\
                                                                                                                                                                 \end{array}
                                                                                                                                                               \right],
\end{equation}
whence
\begin{equation}\label{1_corr_nonstat}
    \ln\det[\Phi(t)\bar{V}(t)]=\frac12\int_0^t d\tau_1d\tau_2(\re+i\sgn(\tau_1-\tau_2)\im)\Sp\big[R_{\tau_2,\tau_1}A(\tau_1)\bar{R}_{\tau_1,\tau_2}\bar{A}(\tau_2)\big]+\cdots,
\end{equation}
in the leading order in $A(\tau)$. If $A(\tau)$ is Hilbert-Schmidt then the operator under the trace sign is trace-class.

For \eqref{Smatr_energy_1} we have
\begin{equation}\label{CA_energy}
    C_{\al\be}=\omega_\al\de_{\al\be}+\{\dot{\bar{\ups}}_\al,\ups_\be\},\qquad A_{\al\be}=-\{\dot{\bar{\ups}}_\al,\bar{\ups}_\be\}.
\end{equation}
It is not difficult to verify that the generating functional for the matrix form of the evolution operator \eqref{evol_symb} is invariant under the replacement
\begin{equation}\label{gauge_transf}
\begin{gathered}
    \ups_\al(\tau)\rightarrow e^{i\vf_\al(\tau)}\ups_\al(\tau),\quad\bar{\ups}_\al(\tau)\rightarrow e^{-i\vf_\al(\tau)}\bar{\ups}_\al(\tau),\quad f_\al(\tau)\rightarrow e^{-i\vf_\al(\tau)}f_\al(\tau),\quad \bar{f}_\al(\tau)\rightarrow e^{i\vf_\al(\tau)}\bar{f}_\al(\tau),\\
    a_\al\rightarrow e^{-i\vf_\al(0)}a_\al, \quad \bar{a}_\al\rightarrow e^{i\vf_\al(t)}\bar{a}_\al.
\end{gathered}
\end{equation}
In particular, $c(t)$ is not changed under the transformations written in the first line of \eqref{gauge_transf}. So the one-loop effective action \eqref{inout_eff_act} is invariant under these transformations.

\section{Density of states and the $S$-matrix}\label{Dens_Stat}

Let us give in this appendix the formal derivation of formula \eqref{spectr_denst_dif} relating the change of the spectral density of a system due to an external field to the $S$-matrix (see, e.g., \cite{Dew_scat,Bordag_scat}).

Let the Hamiltonian $H$ of the system be represented in the form
\begin{equation}
    H=H_0+U.
\end{equation}
Consider the Green functions
\begin{equation}
  G^{-1}(\la)=\la-H,\qquad G^{-1}_0(\la)=\la-H_0,\quad\la\in \mathbb{C}.
\end{equation}
Then, formally, the spectral density with respect to the parameter $\la\in\mathbb{R}$:
\begin{equation}\label{spectr_denst}
  \rho(\la)=\sum_k\de(\la-\la_k)=\Sp\de(\la-H),\qquad \la_k\in\mathrm{spec}\,H,
\end{equation}
where the summation over the discrete spectrum and the integration over the continuous spectrum with an appropriate measure are assumed. The change of the spectral density is written as \cite{Dew_scat}
\begin{equation}
\begin{split}
    \De\rho(\la):=&\sum_k\de(\la-\la_k)-\sum_k\de(\la-\la^0_k)=\Sp[\de(\la-H)-\de(\la-H_0)]=\\
    =&\frac{i}{2\pi}\Sp[(\la-H+i0)^{-1}-(\la-H-i0)^{-1}-(\la-H_0+i0)^{-1}+(\la-H_0-i0)^{-1}]\\
    =&:\frac{i}{2\pi}\Sp[G^+(\la)-G^-(\la)-G_0^+(\la)+G^-_0(\la)].
\end{split}
\end{equation}
Taking into account that, formally,
\begin{equation}
    \Sp G(\la)=-\partial_\la\ln\det G(\la),
\end{equation}
we have
\begin{equation}\label{sectr_dens_ch}
    \De\rho(\la)=-\frac{i}{2\pi}\partial_\la\ln\det\Big\{G^-_0(\la)\big(G^-(\la)\big)^{-1}\Big[G^+_0(\la)\big(G^+(\la)\big)^{-1}\Big]^{-1}\Big\}.
\end{equation}
It is clear that
\begin{equation}
    G_0(\la)G^{-1}(\la)=1-G_0(\la)U.
\end{equation}

Let us define the $T$-matrix as
\begin{equation}
    T(\la)G_0(\la)=UG(\la)\;\Rightarrow\;T(\la)=U(1-G_0(\la)U)^{-1},\qquad (1-G_0(\la)U)^{-1}=1+G_0(\la)T(\la).
\end{equation}
Consequently, we can write
\begin{multline}
    G^-_0(\la)\big(G^-(\la)\big)^{-1}\Big[G^+_0(\la)\big(G^+(\la)\big)^{-1}\Big]^{-1}=[1-G_0^-(\la)U][1-G_0^+(\la)U]^{-1}=\\
    =[1-G_0^+(\la)U]^{-1}-G_0^-(\la)U[1-G_0^+(\la)U]^{-1}
    =1+(G_0^+(\la)-G_0^-(\la))T(\la)=\\
    =1-2\pi i\de(\la-H_0)T(\la)=S(\la),
\end{multline}
where $S(\la)$ is the $S$-matrix off the energy-shell. Then we have
\begin{multline}
    \det\Big\{G^-_0(\la)\big(G^-(\la)\big)^{-1}\Big[G^+_0(\la)\big(G^+(\la)\big)^{-1}\Big]^{-1}\Big\}=\det[\de_{\la''\la'}-2\pi i\de(\la-\la'')T_{\la''\la'}(\la)]=\\
    =\exp\Sp\ln[\de_{\la''\la'}-2\pi i\de(\la-\la'')T_{\la''\la'}(\la)]=1-2\pi iT_{\la\la}(\la)=s(\la),
\end{multline}
where $s(\la)$ is the $S$-matrix on the energy-shell. If the energy level $\la$ is degenerate then the determinant of the matrix $s(\la)$ corresponding to the energy $\la$ stands on the right-hand side of the last equality.

As a result, we obtain
\begin{equation}\label{spec_dens}
    \De\rho(\la)=-\frac{i}{2\pi}\partial_\la\ln s(\la),\quad \la\in \mathbb{R}.
\end{equation}
Notice that $s(\la)$ is not defined in the points of a discrete spectrum where the $S$-matrix possesses poles, and formula \eqref{spec_dens} is to be extended to these points. It follows from the spectral density definition and \eqref{spec_dens} that the function $\ln s(\la)$ should possess the discontinuous jumps equal to $2\pi i$ in the points of the discrete spectrum.

Consider the scattering problem for a one-dimensional Schr\"{o}dinger equation on the whole real axis (see for details, e.g., \cite{Faddeev,FaddZakh,ZMNP}):
\begin{equation}\label{Schr_eq_d1}
    -\psi''+U(x)\psi=\la\psi,\quad x\in(-\infty,+\infty),
\end{equation}
where $U(x)$ is a real continuous function satisfying
\begin{equation}\label{potential_cond}
    \int_{-\infty}^\infty dx(1+|x|)|U(x)|<\infty.
\end{equation}
The solutions of the Schr\"{o}dinger equation \eqref{Schr_eq_d1} corresponding to the same energy $\la=k^2$ and describing the scattering of waves with the momenta $k$ and $-k$,
\begin{equation}\label{psi12}
\begin{alignedat}{2}
    \psi_1(k,x)&\underset{x\rightarrow-\infty}{\rightarrow}e^{ikx}+s_{12}(k)e^{-ikx},&\qquad\psi_1(k,x)&\underset{x\rightarrow+\infty}{\rightarrow}s_{11}(k)e^{ikx},\\
    \psi_2(k,x)&\underset{x\rightarrow-\infty}{\rightarrow}s_{22}(k)e^{-ikx},&\qquad\psi_2(k,x)&\underset{x\rightarrow+\infty}{\rightarrow}e^{-ikx}+s_{21}(k)e^{ikx},
\end{alignedat}
\end{equation}
determine the $S$-matrix
\begin{equation}
    s(\la)=\left[
             \begin{array}{cc}
               s_{11}(k) & s_{12}(k) \\
               s_{21}(k) & s_{22}(k) \\
             \end{array}
           \right].
\end{equation}
It is useful to expand the solutions $\psi_{1,2}$ in the Jost solutions
\begin{equation}\label{psi12_asympt}
\begin{split}
    \psi_1(k,x)&=f^-(k,x)+s_{12}(k)f^-(-k,x)=s_{11}(k)f^+(k,x),\\
    \psi_2(k,x)&=s_{22}(k)f^-(-k,x)=f^+(-k,x)+s_{21}(k)f^+(k,x).
\end{split}
\end{equation}
Whence we deduce the expressions for the elements of the $S$-matrix in terms of the Wronskians
\begin{equation}\label{s_matr_epl}
\begin{alignedat}{2}
    s_{11}(k)&=\frac{W[f^-(-k),f^-(k)]}{W[f^-(-k),f^+(k)]},&\qquad s_{22}(k)&=\frac{W[f^+(k),f^+(-k)]}{W[f^+(k),f^-(-k)]},\\
    s_{12}(k)&=-\frac{W[f^+(k),f^-(k)]}{W[f^+(k),f^-(-k)]},&\qquad s_{21}(k)&=-\frac{W[f^+(-k),f^-(-k)]}{W[f^+(k),f^-(-k)]},
\end{alignedat}
\end{equation}
where for brevity, the argument $x$ of the Jost solutions is omitted. Taking into account that
\begin{equation}\label{wronsk_rel}
    W[f^-(k),f^-(-k)]=2ik,\qquad W[f^+(k),f^+(-k)]=2ik,
\end{equation}
we have
\begin{equation}
    s_{11}(k)=s_{22}(k)=:t(k).
\end{equation}
From \eqref{wronsk_rel} we see that the Wronskian defines a linear symplectic structure on the two-dimensional space of solutions with the bases $\{f^+(k),f^+(-k)\}$ and $\{f^-(k),f^-(-k)\}$:
\begin{equation}\label{simplect_trans}
\begin{split}
    f^+(k,x)=a(k)f^-(k,x)+b(k)f^-(-k,x),\qquad f^+(-k,x)=b(-k)f^-(k,x)+a(-k)f^-(-k,x).
\end{split}
\end{equation}
It follows from \eqref{s_matr_epl} that
\begin{equation}
    a(k)=s_{11}^{-1}(k),\qquad b(k)=s_{12}(k)s_{11}^{-1}(k).
\end{equation}
The conditions of symplecticity of the transform \eqref{simplect_trans} are equivalent to
\begin{equation}\label{untarity}
    t(k)t(-k)+r_-(k)r_-(-k)=1,\qquad t(k)t(-k)+r_+(k)r_+(-k)=1,\qquad t(k)r_+(-k)+r_-(k)t(-k)=0,
\end{equation}
where the notation has been introduced
\begin{equation}
    s(\la)=:\left[
             \begin{array}{cc}
               t(k) & r_-(k) \\
               r_+(k) & t(k) \\
             \end{array}
           \right].
\end{equation}
Besides, it ensues from \eqref{s_matr_epl} that
\begin{equation}\label{unitarity_2}
    \bar{t}(k)=t(-\bar{k}),\qquad \bar{r}_+(k)=r_+(-\bar{k}),\qquad \bar{r}_-(k)=r_-(-\bar{k}).
\end{equation}
For real $k$, the relations \eqref{untarity}, \eqref{unitarity_2} express the unitary of the $S$-matrix. Using these relations, we obtain
\begin{equation}
    \det s(\la)=\frac{t(k)}{t(-k)}=\frac{t(k)}{\bar{t}(\bar{k})}\underset{\la'\in \mathbb{R}}{=}\frac{t(\sqrt{\la_+})}{t(\sqrt{\la_-})},
\end{equation}
where $\la_\pm=\la'\pm i0$, and the cut of the square root is chosen to be along the real positive semiaxis.

As known \cite{Faddeev}, the transition amplitude $t(\sqrt{\la})$ possesses simple poles at $\la<0$ corresponding to the points of a discrete spectrum. It is easy to verify that
\begin{equation}
    \ln\frac{t(\sqrt{\la_+})}{t(\sqrt{\la_-})}
\end{equation}
has the jumps equal to $2\pi i$ in these points and, evidently, coincides with $\ln\det s(\la)$ out of them. So the change of the spectral density \eqref{spec_dens} can be cast into the form \cite{Bordag_scat}
\begin{equation}\label{spectr_denst_dif}
    \De\rho(\la)=-\frac{i}{2\pi}\partial_\la\ln\frac{t(\sqrt{\la_+})}{t(\sqrt{\la_-})}.
\end{equation}
Let $\la_0=\inf\mathrm{supp}\,\De\rho(\la)$, $\la\in \mathbb{R}$. Then, for $\la\not\in[\la_0,+\infty)$, we deduce
\begin{multline}
    \ln\det[G_0(\la)G^{-1}(\la)]=\Sp[\ln(\la-H)-\ln(\la-H_0)]=\int_{\la_0}^\infty d\e\De\rho(\e)\ln(\la-\e)=\\
    =\int_{\la_0}^\infty \frac{d\e}{2\pi i}\partial_\e\ln\frac{t(\sqrt{\e_+})}{t(\sqrt{\e_-})}\ln(\la-\e) =\frac{i}{2\pi}\int_C d\e\partial_\e\ln t(\sqrt{\e})\ln(\la-\e)=-\int_{-\infty}^\la d\e\partial_\e\ln t(\sqrt{\e})=\ln t^{-1}(\sqrt{\la}),
\end{multline}
where the contour $C$ passes from $+\infty$ a little bit higher than the real axis, encircles counter-clockwise the point $\la_0$, and then goes to $+\infty$ a little bit lower than the real axis. The branch of the function $\ln(\la-\e)$ is taken with the cut along the real positive semiaxis. In the penultimate equality, the contour $C$ is deformed to the contour encompassing the branch cut of $\ln(\la-\e)$. At that, it is taken into account that $\ln t(\sqrt\la)$ is holomorphic on the physical sheet out of the ray $[\la_0,+\infty)$ and \cite{Faddeev}
\begin{equation}\label{t_asympt}
    t(\sqrt{\la})\underset{|\la|\rightarrow\infty}{\rightarrow}1+O(1/\sqrt{\la}).
\end{equation}
As a result, we have
\begin{equation}
    \det[G_0(\la)G^{-1}(\la)]=t^{-1}(\sqrt{\la}).
\end{equation}
The expression on the left-hand side is the Fredholm determinant \cite{Newton_scat.12} (cf. \eqref{sectr_dens_ch} and \eqref{spectr_denst_dif}).

\section{Analytical properties of the adiabatic invariant}\label{Adiabatic_Invariant}

Let us study some analytical properties of the adiabatic invariant,
\begin{equation}\label{adiab_inv}
    I(\la)=2\int_{x_1(\la)}^{x_2(\la)}dx\sqrt{\la-V(x)},\qquad V(x_{1,2}(\la))=\la,\quad V'(x_{1,2}(\la))\neq0,
\end{equation}
as a function of the complex variable $\la$, the principal branch of the square root being taken in \eqref{adiab_inv}. Suppose that the function $V(x)$ is analytic in the vicinity of the real axis and possesses only pole singularities in the complex plane. Assume that, at infinity, the potential behaves as
\begin{equation}\label{potent_asympt}
    V(x)\rightarrow c_+\;\;\text{for}\;\;|x|\rightarrow\infty, \;\re x>0;\qquad V(x)\rightarrow c_-\;\;\text{for}\;\;|x|\rightarrow\infty, \;\re x<0,
\end{equation}
and there is the interval of values of $\la\in \mathbb{R}$ such that
\begin{equation}
    \la-V(x)>0,\quad\forall x\in(x_1(\la),x_2(\la)).
\end{equation}
Then $I(\la)$ can be continued analytically to the complex plane so that
\begin{equation}
    \bar{I}(\la)=I(\bar{\la}).
\end{equation}
The turning points $x_{1,2}(\la)$ also obey the Schwarz symmetry principle
\begin{equation}\label{roots_schw_p}
    \bar{x}_{1,2}(\la)=x_{1,2}(\bar{\la}).
\end{equation}

The singularities of the function $I(\la)$ appear only in the case when
\renewcommand{\theenumi}{\alph{enumi}}
\renewcommand{\labelenumi}{\theenumi)}
\begin{enumerate}
  \item $x_{1}(\la)$ and/or $x_{2}(\la)$ go to infinity;
  \item The turning points different from $x_{1,2}(\la)$ pitch the integration contour connecting $x_{1}(\la)$ with $x_{2}(\la)$;
  \item The turning points different from $x_{1,2}(\la)$ approach one (or both) the points $x_{1,2}(\la)$.
\end{enumerate}
\renewcommand{\theenumi}{\arabic{enumi}}
\renewcommand{\labelenumi}{\theenumi.}
At the same time, when the turning points $x_{1,2}(\la)$ coalesce, the function $I(\la)$ does not have a singularity. This is easy to see if one writes \eqref{adiab_inv} as the integral over the contour enclosing the branch points $x_{1,2}(\la)$ of the integrand. The cut of the square root is then taken to be some curve connecting $x_1(\la)$ with $x_2(\la)$. The singularities of $I(\la)$ of the types $b)$ and $c)$ occur only when this contour is pitched.

Let us find the values of $\la$ where these singularities arise and determine the behaviour of the function $I(\la)$ near these points. It follows from the asymptotics \eqref{potent_asympt} that the case $a)$ is realized at $\la=c_{\pm}$. We consider the behaviour of $I(\la)$ only in the neighbourhood of the point $c_+$. As for the point $\la=c_-$, the consideration is completely analogous. We restrict our consideration to two the most frequent types of behaviour of the potential at infinity
\begin{equation}\label{potent_asympt_epl}
    i)\;V(x)\underset{x\rightarrow+\infty}{\rightarrow} c_++a_+e^{-\be_+ x},\qquad ii)\;V(x)\underset{x\rightarrow+\infty}{\rightarrow} c_++a_+x^{-N_+},\;N_+>2.
\end{equation}
Let the turning point $x_2(\la)\rightarrow+\infty$ at $\la\rightarrow c_+$, and $x_1(\la)$ has a finite limit in this case. Then we can write for a period
\begin{equation}\label{period}
    \tau(\la):=I'(\la)=\int_{x_1(\la)}^{M}\frac{dx}{\sqrt{\la-V(x)}}+\int_{M}^{x_2(\la)}\frac{dx}{\sqrt{\la-V(x)}},\quad M\in(x_1,x_2),
\end{equation}
where the point $M$ is chosen sufficiently close to the point $x_2$ such that the potential $V(x)$ is well approximated by the asymptotics \eqref{potent_asympt_epl} on the interval $[M,x_2]$. We are interested only in the leading asymptotics of $\tau(\la)$ singular at $\la\rightarrow c_+$. The first integral is finite at $\la\rightarrow c_+$. In order to evaluate the second integral, we use the asymptotics \eqref{potent_asympt_epl}.

In the case $i)$, we pass to the integration variable $V$. Then the second integral is written as
\begin{equation}
    \int_{V(M)}^\la\frac{dV}{\be_+(c_+-V)\sqrt{\la-V}}=\frac{2}{\be_+\sqrt{c_+-\la}}\arctg\sqrt{\frac{\la-V(M)}{c_+-\la}}\underset{\la\rightarrow c_+}{\rightarrow}\frac{\pi}{\be_+\sqrt{c_+-\la}}+o\big((c_+-\la)^{-1/2}\big).
\end{equation}
Consequently, in this case, the adiabatic invariant becomes
\begin{equation}\label{Ii}
    I(\la)=2\int_{x_1}^{+\infty}dx\sqrt{c_+-V(x)}-\frac{2\pi}{\be_+}\sqrt{c_+-\la}+o\big((c_+-\la)^{1/2}\big).
\end{equation}
In the case $ii)$, after a suitable change of the integration variable, the second integral takes the form
\begin{equation}
    \frac{1}{\sqrt{c_+-\la}}\Big(\frac{a_+}{\la-c_+}\Big)^{1/{N_+}}\int_{M((\la-c_+)/a_+)^{1/N_+}}^1\frac{dx x^{N_+/2}}{\sqrt{1-x^{N_+}}} \underset{\la\rightarrow c_+}{\rightarrow}  \frac{\Ga(1/2+1/N_{+})}{\Ga(1/N_+)}\frac{\sqrt{\pi}(-a_+)^{1/N_+}}{(c_+-\la)^{1/2+1/N_+}}.
\end{equation}
As a result, the expansion of the adiabatic invariant in this case reads as
\begin{equation}\label{Iii}
    I(\la)=2\int_{x_1}^{+\infty}dx\sqrt{c_+-V(x)}+\sqrt{\pi}\frac{\Ga(1/N_+-1/2)(c_+-\la)^{1/2-1/N_+}}{\Ga(1/N_+)(-a_+)^{-1/N_+}}+o\big((c_+-\la)^{1/2-1/N_+}\big).
\end{equation}
The principal branches of the multivalued functions are taken in \eqref{Ii}, \eqref{Iii}. If $c_\pm=c$ and $x_1\rightarrow-\infty$, and $x_2\rightarrow+\infty$ then in formulas \eqref{Ii}, \eqref{Iii} the lower limit $x_1$ is replaced by $-\infty$ and the additional contributions to the leading correction to the asymptotics arise from the point $x_1$. These corrections have the same form as the corrections from the point $x_2$, but with the replacement of indices ``$+$'' by ``$-$''.

Now we turn to the cases $b)$ and $c)$. In a general position, the contour is pitched by the two coalescing roots of the equation
\begin{equation}
  V(x)=\la,
\end{equation}
defining the turning points, i.e., in this point $x_0\in \mathbb{C}$, the potential $V(x_0)$ possesses an extremum. The value of $\la$ where this singularity of the function $I(\la)$ occurs is equal to the value of $V(x)$ at the extremum point. Without loss of generality, we can take $x_0=0$ and $V(x_0)=0$. Inasmuch as only two points coalesce, $V''(x_0)=:2\al\neq0$. Then in the vicinity of the point $\la=0$, we have
\begin{equation}
  y_{1,2}=\pm\sqrt{\frac{\la}{\al}}+O(\la),
\end{equation}
where $y_{1,2}$ are the coalescing turning points. In the case $c)$, one of these points coincide with $x_1$ or $x_2$. In the case $b)$, the period is given by
\begin{equation}\label{period_1}
  \tau(\la)=\int_{x_1}^{x_2}dx\Big(\frac{1}{\sqrt{\la-V(x)}}+\frac{\la}{y_1y_2\sqrt{(x-y_1)(x-y_2)}}\Big)-\frac{\la}{y_1y_2}\int_{x_1}^{x_2}\frac{dx}{\sqrt{(x-y_1)(x-y_2)}}.
\end{equation}
We can put $\la=0$ in the first integral, while the second integral is readily calculated. As a result, for $\la\rightarrow0$,
\begin{equation}
  \tau(\la)=\int_{x_1(0)}^{x_2(0)}dx\Big(\frac{1}{\sqrt{-V(x)}}+\frac{\al}{\sqrt{x^2}}\Big)-\al\ln\frac{4\al x_1(0)x_2(0)}{\la}+o(1).
\end{equation}
Integrating this expression over $\la$, we derive the leading terms of the expansion of $I(\la)$ in the vicinity of the singular point $\la=0$. In the case $c)$, we suppose that $x_1=y_2$. Then, writing the period in the form \eqref{period_1}, we have, for $\la\rightarrow0$,
\begin{equation}
  \tau(\la)=\int_{0}^{x_2}dx\Big(\frac{1}{\sqrt{-V(x)}}+\frac{\al}{x}\Big)-\al\ln\Big(2x_2(0)\sqrt{\frac{\al}{\la}}\Big)+o(1),
\end{equation}
whence the expansion of $I(\la)$ follows in this case. As we see, in a general position in the cases $b)$ and $c)$, the adiabatic invariant possesses the branch points of the form $\la\ln\la$.

Let us obtain the leading asymptotics of $I(\la)$ for $|\la|\rightarrow\infty$. At large $\la$, the turning points $x_{1,2}(\la)$ tend to the singular points of the potential $V(x)$ where its values become infinite. Notice that $x_{1,2}(\la)$ do not go to infinity at $|\la|\rightarrow\infty$ in virtue of the condition \eqref{potent_asympt}. Let $x_{1,2}$ tend to the singular points $y_{1,2}$ of the potential $V(x)$, respectively. In the neighbourhood of these points, the potential behaves as
\begin{equation}\label{potent_pole}
  V(x)\approx a_1/(x-y_1)^{N_1}+O((x-y_1)^{1-N_1}),\qquad V(x)\approx a_2/(x-y_2)^{N_2}+O((x-y_2)^{1-N_2}).
\end{equation}
Then the period can be represented as
\begin{equation}\label{adiab_inv_split}
  I(\la)=2\Big(\int_{x_1}^{M_1}+\int^{M_2}_{M_1}+\int_{M_2}^{x_2}\Big)dx\sqrt{\la-V(x)},
\end{equation}
where the points $M_1$ and $M_2$ are chosen sufficiently close to the points $y_1$ and $y_2$ such that, on the segments $[x_1,M_1]$ and $[M_2,x_2]$, the function $V(x)$ is well approximated by the leading term of the expansion \eqref{potent_pole}. The potential is bounded on the interval $[M_1,M_2]$ and so the leading asymptotics at $|\la|\rightarrow\infty$ is easily found
\begin{equation}
  2\int^{M_2}_{M_1}dx\sqrt{\la-V(x)}= 2\sqrt{\la}\int^{M_2}_{M_1}dx+o(\sqrt{\la}).
\end{equation}
In the two rest integrals, one can replace $V(x)$ by its expansion in the vicinity of the singular points. After a suitable change of the variables, the first integral in \eqref{adiab_inv_split} becomes
\begin{equation}
  2\sqrt{\la}\Big(\frac{a_1}{\la}\Big)^{1/N_1}\int_1^{(M_1-y_1)(\la/a_1)^{1/N_1}}dx\sqrt{1-x^{-N}}= 2\sqrt{\la}\int^{M_1}_{y_1}dx+o(\sqrt{\la}).
\end{equation}
The same applies to the second point. As a result, we obtain the leading asymptotics at $|\la|\rightarrow\infty$:
\begin{equation}\label{adiab_inv_asymp}
  I(\la)=2\sqrt{\la}(y_2-y_1)+o(\sqrt{\la}),
\end{equation}
where the principal branch of the square root is chosen. Note that, owing to the property \eqref{roots_schw_p}, this asymptotics obeys the Schwarz symmetry principle as it should be.

Let us discuss in conclusion the analytical properties of the adiabatic invariants corresponding to the infinite motion, viz., when one or both of the turning points $x_{1,2}$ in the integral \eqref{adiab_inv} are replaced by the fixed values $\pm L/2$ corresponding in the quantum case to the boundaries of the segment where the Sturm-Liouville problem is solved. It is not difficult to check that all the above properties also take place for these adiabatic invariants. The only exception is that there is no singularity of $I(\la)$ when some turning point tends to the points $\pm L/2$.

\section{Regularization and renormalization}\label{Reg_Renorm}

Let us give in this appendix the main definitions necessary for us that are related to the regularization and renormalization procedures in QFT. For more details see \cite{Collins.12,Zavyalov,Smirnov}.
\begin{defn}\label{reg_defn}
  The regularization of quantum field theory with the classical action $S$ is such a deformation $S_\La$ of its action that
\renewcommand{\theenumi}{\alph{enumi}}
\renewcommand{\labelenumi}{\theenumi)}
  \begin{enumerate}
    \item In the regularization removal limit $S_\La\rightarrow S$ for $\La\rightarrow+\infty$;
    \item The terms of the perturbation series given by the Feynman diagrams converge for $S_\La$.
  \end{enumerate}
\renewcommand{\theenumi}{\arabic{enumi}}
\renewcommand{\labelenumi}{\theenumi.}
\end{defn}

One can conventionally distinguish the two classes of regularizations:
\begin{enumerate}
    \item The regularizations saving as many symmetries of the initial classical theory as possible;
    \item The regularizations saving a unitarity of the theory.
\end{enumerate}
To the first class, one can attribute the Pauli-Villars regularization and the regularization by higher derivatives. These regularizations keep the Lorentz-invariance intact, but violate the unitarity of a theory. As for the second class, these are the cut-off regularizations and the regularizations modifying the particles' dispersion laws at high energies.

As a rule, it is impossible in relativistic QFT without extended symmetries (such as, for example, the superconformal one) in the four-dimensional spacetime to provide such a regularization of the theory that saves both the Lorentz-invariance and the unitarity. It can be seen from the K\"{a}ll\'{e}n-Lehmann representation of the complete propagator that implies that the asymptotic behaviour of the complete propagator at large momenta is the same as for the bare propagator (see, e.g., \cite{WeinbergB.12}), or from the Weinberg theorem (see, e.g., \cite{Collins.12}) describing the asymptotic behaviour of an arbitrary diagram at large momenta.

For the regularizations schemes preserving unitary but spoiling Lorentz-invariance, one can hope, at least, that, before taking the regularization removal limit, the asymptotic perturbation series resummed appropriately (by the Borel method or more powerful methods) defines the unitary operator relating the operators in the interaction and the Heisenberg pictures. This is \textit{a priori} impossible for the Lorentz-invariant regularizations in virtue of the Haag theorem \cite{Barton,StreatWight}.

\begin{defn}
  The perturbative procedure of addition of the counterterms to the regularized action $S_\La$ providing
\renewcommand{\theenumi}{\alph{enumi}}
\renewcommand{\labelenumi}{\theenumi)}
  \begin{enumerate}
    \item The finiteness of the perturbation series terms in the regularization removal limit $\La\rightarrow+\infty$;
    \item The concordance of QFT with experiments: the physical normalization conditions on the particle masses and coupling constants, the residues of the propagators at poles, etc.;
  \end{enumerate}
\renewcommand{\theenumi}{\arabic{enumi}}
\renewcommand{\labelenumi}{\theenumi.}
  is called renormalization.
\end{defn}
As a rule, the counterterms are chosen to save as many classical symmetries of the action $S$ as possible in QFT in the regularization removal limit provided these symmetries are not broken by other mechanisms which are not concerned with the regularization.

The form of the counterterms is subject to certain restrictions -- the counterterms should be polynomial in the fields and derivatives, viz., they should be the sum (possibly infinite) of the monomials
\begin{equation}\label{countt}
    \partial_i^{[k]}\phi^{[n]},
\end{equation}
where $k$ derivatives are ``distributed'' somehow among $n$ fields (in the case of gravity $\phi=g_{\mu\nu}-\eta_{\mu\nu}$). One can give two arguments, at least, in favor of this restriction. The first argument is that only such counterterms appear not to violate the cluster decomposition principle (\cite{WeinbergB.12}, Sect. 4.4), although this statement is not a theorem. The second argument consists in that such counterterms are sufficient to cancel out all the ultraviolet divergencies of any relativistic QFT (see for details, e.g., \cite{Collins.12}). For renormalizable models, it is also demanded that the mass dimension of a counterterm does not exceed the spacetime dimension. As for non-renormalizable models, this limitation is absent. It is the restriction \eqref{countt} on the form of counterterms that allows one to make non-trivial predictions in the non-renormalizable QFT. Otherwise, any term in the effective action of such a theory can be canceled out by a counterterm.

After renormalization, the perturbation series for the $n$-point Green functions does not depend on the regularization scheme. For the regularizations of the first type, the Lorentz-invariance is preserved non-perturbatively, while the unitarity is restored perturbatively in the regularization removal limit. For the regularizations of the second type, the theory is unitary, but the Lorentz-invariance (and, possibly, other symmetries) is restored perturbatively at every order of the perturbation theory.

Notice that the dimensional and analytical regularizations (see for details, e.g., \cite{Collins.12,Zavyalov}) are not the regularizations in the sense of the definition \ref{reg_defn} since: a) They are not the deformations of the initial classical action and are applied directly to the analytic expressions entering the perturbation series; b) In the regularization removal limit, these schemes do not reproduce the power-like divergencies, i.e. in this limit, the regularized expression does not tend to the initial ill-defined analytic expression for the term of the perturbation series. These regularization schemes can be considered as a combination of the regularization (for instance, by Pauli-Villars) and renormalization procedures in the sense defined above. The dimensional and analytic regularizations are particulary convenient when the logarithmic divergencies of a theory are only needed to be traced. There exist other regularization procedures for the free fields on a curved background (see, e.g., \cite{WaldReg}) which, apparently, are not reduced to the above scheme.

\paragraph{Acknowledgments.}

The work of the authors is partly supported by the RFBR grant No 16-02-00284.

\end{document}